\documentclass[sn-mathphys]{sn-jnl}
\jyear{2022}

\usepackage{physics}
\usepackage{amsthm,thmtools}
\usepackage{amssymb}

\theoremstyle{thmstyleone}
\newtheorem{theorem}{Theorem}
\newtheorem{corol}{Corollary}
\newtheorem{lemma}{Lemma}
\newtheorem{definition}{Definition}
\newtheorem{remark}{Remark}

\renewcommand{\vu}[1]{\boldsymbol{#1}}

\newcommand{\upd}[1]{^\mathrm{#1}}
\newcommand{\ind}[1]{_\mathrm{#1}}

\newcommand{\banachspace}{B(\mathbb{R})}
\newcommand{\banachspaceT}{B\ind{T}(\mathbb{R})}
\newcommand{\banachspaceplus}{B_+(\mathbb{R})}

\title[Existence and Uniqueness for GHD equation]{Existence and Uniqueness of Solutions to the Generalized Hydrodynamics Equation}

\author*[1]{\fnm{Friedrich} \sur{H\"ubner}}\email{friedrich.huebner@kcl.ac.uk}

\author[1]{\fnm{Benjamin} \sur{Doyon}}\email{benjamin.doyon@kcl.ac.uk}

\affil*[1]{\orgdiv{Department of Mathematics}, \orgname{King’s College London}, \orgaddress{\street{Strand}, \city{London}, \postcode{WC2R 2LS}, \country{UK}}}

\abstract{
The generalized hydrodynamics (GHD) equation is the equivalent of the Euler equations of hydrodynamics for integrable models. Systems of hyperbolic equations such as the Euler equations usually develop shocks and are plagued by problems of uniqueness. We establish for the first time the existence and uniqueness of solutions to the full GHD equation and the absence of shocks, from a large class of initial conditions with bounded occupation function. We assume only absolute integrability of the two-body scattering shift. In applications to quantum models of fermionic type, this includes all commonly used physical initial states, such as locally thermal states and zero-entropy states. We show in particular that differentiable initial conditions give differentiable solutions at all times and that weak initial conditions such as the Riemann problem have unique weak solutions which preserve entropy. For this purpose, we write the GHD equation as a new fixed-point problem (announced in a companion paper). We show that the fixed point exists, is unique, and is approached, under an iterative solution procedure, in the Banach topology on functions of momenta.
}


\begin{document}

\maketitle

\tableofcontents

\section{Introduction}
Generalized hydrodynamics (GHD) was introduced as a hydrodynamic theory to describe the large-scale behaviours of integrable systems~\cite{PhysRevX.6.041065,bertini2016transport} (for overviews see~\cite{BenGHD,Bastianello_2022,ESSLER2023127572}). Integrable models are special fine-tuned systems which possess a large number of local conserved quantities, proportional to the system's size, and a mathematical structure which often allows one to find exact expressions for physical quantities of interest. Integrable models exist in all dimensions, but, up to now, the most powerful results are found in one-dimensional systems, on which we concentrate in this paper.

In typical non-integrable systems, it is believed that only a few local conserved quantities exist. At large scales, under the assumption of local thermodynamic relaxation, the dynamics of such systems can be described by the corresponding continuity equations: these are the conventional Euler-scale equations of hydrodynamics \cite{spohn1991large}.
For integrable models, one adapts these principles to all available conserved quantities, of which there are infinitely many in the large-scale limit: this is GHD \cite{PhysRevX.6.041065,bertini2016transport}. The resulting infinite system of conservation equations can be recast into an equation describing the evolution of quasi-particle densities in a certain phase space. Each quasi-particle propagates with an effective velocity, a modification of its bare velocity due to the interaction with other quasi-particles. Quasi-particles may be interpreted as ``stable excitations", and are in correspondence with the asymptotically stable objects of the scattering theory of the model. In fact, the GHD equations are an extension to quantum and classical many-body integrable systems of all types, of the kinetic equations for soliton gases \cite{el2003thermodynamic} (see e.g.~the reviews \cite{el2021soliton,suret2023soliton}), where the quasi-particles are the solitons of an integrable partial differential equations. GHD equations have been extended to include hydrodynamic diffusion~\cite{hübner2024diffusivehydrodynamicslongrangecorrelations,PhysRevLett.121.160603,10.21468/SciPostPhys.6.4.049} and dispersion \cite{DeNardisHigher_2023}, external potentials~\cite{10.21468/SciPostPhys.2.2.014,PhysRevLett.122.240606} and much more~\cite{PhysRevLett.123.130602,10.21468/SciPostPhys.15.4.136,PhysRevB.102.161110,biagetti2024generalisedbbgkyhierarchynearintegrable}. For simplicity we will drop the prefix ``quasi", as the specific physical interpretation does not concern us.

The GHD equation describes the evolution of the density $\rho_{\rm p}(t,x,p)$ of particles per unit momentum and per unit distance, at momentum $p$, position $x$ and time $t$:
\begin{align}
    \partial_t \rho_{\rm p}(t,x,p) + \partial_x (v\upd{eff}(t,x,p)\rho_{\rm p}(t,x,p)) = 0.\label{equ:intro_GHD_conservation}
\end{align}
The effective velocity solves the following linear integral equation:
\begin{align}
    v\upd{eff}(t,x,p) &= v(p) + \int\dd{q}\varphi(p,q) \rho_{\rm p}(t,x,q) (v\upd{eff}(t,x,q)-v\upd{eff}(t,x,p)).\label{equ:intro_GHD_conservation_velocity}
\end{align}
The bare velocity $v(p)$ and scattering shift $\varphi(p,q)$ are model specific. The later is the displacements particles incur at two-body scattering events, and fully encodes the interaction.
We choose the momentum variable $p$ to take values in $\mathbb R$; in general, $p$ is a spectral parameter and may take values in a more general manifold, see e.g.~\cite{doyon2018exact} and reviews in \cite{Bastianello_2022}. We will state our precise mathematical setup for \eqref{equ:intro_GHD_conservation}, \eqref{equ:intro_GHD_conservation_velocity} in Section \ref{sectsummary}.

An important question is as to the solution of the GHD equation for $\rho(t,x,p)$, seen as an Cauchy problem with initial condition $\rho(0,x,p)$. For conventional Euler-type equations it is a non-trivial problem to establish the existence and uniqueness of the solution even for finite times; typically shocks eventually develop, where the solution becomes weak and is not unique anymore, see for instance \cite{Bressan2013}. But the GHD equation is very special, as it satisfies a “linear degeneracy” condition \cite{FERAPONTOV1991112,el2011kinetic,pavlov2012generalized,BenGHD}. In finite systems of hyperbolic equations, this condition is well known to preclude the formation of true shocks \cite{lindeg1,lindeg2}. Restricting to initial conditions supported on a finite number of values of $p$, the GHD equation reduces to such a finite system that is of semi-Hamiltonian type, and hodographic and other methods can be used \cite{FERAPONTOV1991112,el2011kinetic,pavlov2012generalized,PhysRevLett.119.220604}. One may therefore expect similar statements to hold for the full GHD equation, but in this infinite-dimensional case we are not aware of a general theory (although the Hamiltonian structure has recently been developed \cite{vergallo2024hamiltonian,bonnemain2024hamiltonian}). Yet, it has been observed that shocks do not appear in the partitioning protocol \cite{PhysRevX.6.041065}, and similar no-shock effects are seen at zero entropy \cite{PhysRevLett.119.195301}.

The GHD equation was reformulated in various ways as systems of nonlinear integral equations, where time appears explicitly -- thus, the Cauchy problem is recast into the problem of determining the space of solutions to these nonlinear integral equations. One reformulation ~\cite{PhysRevLett.119.220604} (supplementay material -- Continuum Limit) gives $\rho_{\rm p}(t,x,\cdot)$ for any $(x,t)$ as the solution to a set of integral equations with $\rho_{\rm p}(0,\cdot,\cdot)$ as an input. However, it requires special initial conditions, those that are invertible as functions of $x$. Another \cite{DOYON2018570} gives $\rho_{\rm p}(t,\cdot,\cdot)$ on any time slice $t$ as a solution to a different set of integral equations, with $\rho_{\rm p}(0,\cdot,\cdot)$ as an input. It only requires that the state be asymptotically flat in space and is based on the GHD characteristics. It can be interpreted as the result of a state-dependent change of spatial coordinates that maps the GHD equation to the Liouville phase-space conservation equation. Numerical analysis shows that a recursive solution to the integral equations converges in many cases \cite{DOYON2018570,10.21468/SciPostPhysCore.3.2.016}. However, rigorous analyses of the integral equations of \cite{PhysRevLett.119.220604} and  \cite{DOYON2018570} have not been given.

The aim of this paper is to rigorously establish the existence and uniqueness of solutions to the Cauchy problem for the GHD equation. This is based on a novel reformulation of the GHD equation as a system of integral equations for $\rho_{\rm p}(t,x,\cdot)$ taking $\rho_{\rm p}(0,\cdot,\cdot)$ as an input, which is explained in detail in our companion paper~\cite{hübner2024newquadraturegeneralizedhydrodynamics}. It is a fixed-point problem on the Banach space of bounded functions of momentum, which is contracting under certain general conditions. Therefore, the solution for $\rho_{\rm p}(t,x,\cdot)$, given $\rho_{\rm p}(0,\cdot,\cdot)$ and $x,t$, exists and is unique. The derivation uses similar ideas as the GHD characteristics approach, but goes further by analyzing ``height fields" associated to the conserved densities. We find that
\begin{itemize}
    \item The solution to the GHD equation exists at all times and is unique for a large class of initial conditions.
    \item $r$ times differentiable initial conditions give rise to $r$ times differentiable solutions, thus {\em no discontinuity of any type appear for all times}.
    \item Discontinuous initial conditions have a {\em unique weak solution}. A weak solution necessitates a reformulation of the differential equation \eqref{equ:intro_GHD_conservation} in order to make sense of the derivatives. We use a simple re-reformulation in terms of height fields, tailored to systems of conservation laws.
    \item Entropy is conserved throughout time, including when discontinuities are present.
\end{itemize}
In particular, from the above, no shock may appear or be sustained. This holds even with a discontinuous initial condition. Discontinuities may be carried over time, but they are contact discontinuities instead of shocks (as first proposed in \cite{PhysRevX.6.041065}), as entropy is conserved. Further, the resulting weak solution is unique: there is no need to appeal to extra physical conditions to establish how they propagate \cite{Bressan2013}.

The class of GHD equation we consider includes that expected to describe the (repulsive) Lieb-Liniger gas \cite{lieb1963} that models one-dimensional cold atomic gases realised in experiments \cite{Bouchoule_2022}, but excludes those from most soliton gases \cite{el2021soliton,suret2023soliton}.

The paper is structured as follows. In Section \ref{sectsummary} we give the precise statement of our main theorems concerning the existence and uniqueness of the solution to the GHD equations, explicitly stating the conditions and framework. In Section \ref{sectoverview} we present a non-rigorous discussion of GHD and the new fixed-point problem. This serves to explain the main aspects of the statements and proofs. In Section \ref{sec:basics} we expand on the full mathematical setup and establish basic properties of the dressing operation. In Section \ref{sec:fixedpoint} we construct the fixed-point problem and proves its main properties. In Section \ref{sec:existenceuniqueness} we use these results to show the main theorems of existence and uniqueness. In Section \ref{sec:consequences} we establish some basic consequences of the results, including entropy conservation, and finally in Section \ref{sec:conclu} we give our concluding remarks.

\section{Setup and summary of the main results}\label{sectsummary}
We first give an overview over our setup and main results.

\subsection{Setup and assumptions}\label{ssectassumptions}

For simplicity of the discussion, momentum variables $p,q,\ldots$ take values in $\mathbb R$ on which we take the Lebesgue measure -- and we will omit the integration range when the integral is on $\mathbb R$. However there is no difficulty in adapting our results to any other measure space instead of $(\dd{p},\mathbb R)$, including cases with many particle species or with a Brillouin zone (a finite domain for the momentum variables) -- see discussions in \cite{doyon2018exact,10.21468/SciPostPhys.6.4.049}. We also consider the problem on spatial variables $x\in\mathbb R$; it is not difficult to adapt our discussion to other situations such as finite intervals with periodic conditions.

We assume given a function $\varphi:\mathbb R^2\to\mathbb R$ (scattering shift) and a function $v:\mathbb R\to\mathbb R$ (velocity), together characterising the model; as well as a function $\rho\ind{p}(0,\cdot,\cdot):\mathbb R^2\to\mathbb R$, the initial particle density. We also assume, without further stating it, that $\varphi$ is Lebesgue measurable, that $\rho\ind{p}(0,x,\cdot)$ is Lebesgue measurable for every $x\in\mathbb R$, and that $\varphi(p,q)$ and $\varphi(p,q)\rho_{\rm p}(0,x,q)$ are absolutely integrable (in $\mathbb R$) on $q\in\mathbb R$ for every $p,x\in\mathbb R$. We also require throughout that
\begin{equation}\label{equ:supvarphi}
    \sup_{p\in\mathbb R} \int\dd{q}\abs{\varphi(p,q)}<\infty.
\end{equation}

The non-trivial conditions on the initial data are most easily expressed in terms of the initial data, $t=0$, for the so-called occupation function $n(t,x,p)$. The occupation function, for given $\rho\ind{p}(t,x,p)$ such that $\varphi(p,q)\rho_{\rm p}(t,x,q)$ is absolutely integrable on $q\in\mathbb R$, is defined by: 
\begin{align}
    n(t,x,p) &= \frac{\rho\ind{p}(t,x,p)}{\rho\ind{s}(t,x,q)}
    \label{equ:summaryn}
\end{align}
where we define the ``state density"
\begin{equation}
    \rho\ind{s}(t,x,p) = \frac1{2\pi}\Big(
    1+\int_{\mathbb R}\dd{q} \varphi(p,q)\rho\ind{p}(t,x,q)\Big).
    \label{equ:summaryrhos}
\end{equation}
We will denote the initial data of the occupation function as $n_0(x,p) = n(0,x,p)$.

Our main assumptions are as follows; these should be seen as constraints on the initial particle density $\rho\ind{p}(0,x,p)$, which depend on the scattering shift $\varphi(p,q)$ and the velocity $v(p)$.
\medskip

\noindent{\bf Assumptions.} {\em The initial occupation function is bounded and non-negative, $n_0(x,p)\geq 0$ for all $x,p$. Further, either:
\begin{align}
    \varphi(p,q)&\geq 0\ \forall\;p,q\in\mathbb R \quad\mbox{or}\quad  \varphi(p,q)\leq 0\ \forall\;p,q\in\mathbb R,\nonumber\\
    &\mbox{and}\quad \sup_{p\in\mathbb R} \int\dd{q}\abs{\varphi(p,q)}\sup_{x \in\mathbb{R}}n_0(x,q)<2\pi;\label{equ:summary_bound_n}
\end{align}
or
\begin{align}
    \sup_{p\in\mathbb R} \int\dd{q}\abs{\varphi(p,q)}\sup_{x \in\mathbb{R}}n_0(x,q)<\pi.\label{equ:summary_bound_n_2}
\end{align}
Finally,
\begin{equation}\label{equ:summary_bound_v}
    \sup_{(x,p)\in\mathbb R^2} \abs{v(p)n_0(x,p)} <\infty.
\end{equation}
}

We will show that given any (Lebesgue measurable) non-negative $n_0(x,p)$ satisfying \eqref{equ:summary_bound_n} or \eqref{equ:summary_bound_n_2}, there are unique (finite) $\rho\ind{p}(0,x,p)$, $\rho\ind{s}(0,x,p)$, and that these satisfy $\rho\ind{p}(0,x,p)\geq0,\ \rho\ind{s}(0,x,p)>0$ for all $x,p$ (Lemma \ref{lem:initialdensities}). Thus, in particular, the full space of functions $n_0(x,p)$ as specified in the Assumptions is available. We note that the inequality $\rho\ind{s}(0,x,p)>0$ is essential for the mathematical proof, while the inequality $\rho\ind{p}(0,x,p)\geq 0$ (related to the condition $n_0(x,p)\geq 0$) is not. The latter is imposed for simplicity, as in all currently known physical applications of GHD, $\rho\ind{p}(0,x,p)$ is a density, hence non-negative. Without this condition we would have to replace $\varphi(p,q)$ by  $\varphi(p,q)n_0(x,q)$ in the first inequalities in \eqref{equ:summary_bound_n} (which would have to hold for all $x\in\mathbb R$), and use $\abs{n_0(x,p)}$ in the last inequality in \eqref{equ:summary_bound_n} and \eqref{equ:summary_bound_n_2}.

Because the occupation function is a Riemann invariant for the GHD equations (see e.g.~\cite{BenGHD}), then, if the characteristics do not close upon themselves, $n(t,x,p)$ satisfies the same condition for all times $t\in\mathbb R$. In fact, we will show that
\begin{align}
    & \rho\ind{p}(t,x,q)\geq 0,\quad \rho\ind{s}(t,x,q)>0\quad \forall \;t,x,p\in\mathbb R
    \label{equ:summary_bound_rho_t}
\end{align}
and either (case \eqref{equ:summary_bound_n})
\begin{align}
    \sup_{p\in\mathbb R} \int\dd{q}\abs{\varphi(p,q)}\sup_{(x,t) \in\mathbb{R}^2}n(t,x,q)<2\pi\label{equ:summary_bound_n_t}
\end{align}
or (case \eqref{equ:summary_bound_n_2})
\begin{align}
    \sup_{p\in\mathbb R} \int\dd{q}\abs{\varphi(p,q)}\sup_{(x,t) \in\mathbb{R}^2}n(t,x,q)<\pi\label{equ:summary_bound_n_t_2}.
\end{align}

Clearly, in general, because of the assumptions above, our results apply to states of low enough density. However, as it turns out, in some quantum models, including the experimentally relevant Lieb-Liniger model, $\varphi(p,q)>0$ and the states which satisfy the bound (\ref{equ:summary_bound_n}) comprise all locally thermal states, locally zero-entropy states, as well as generic local generalised Gibbs ensembles that have been considered up to now in the literature (as far as we are aware of). Therefore the following theorems will in particular apply to those models in quite some generality. 
\subsection{Main results}
Our first result establishes the existence of a weak solution to the GHD equation \eqref{equ:intro_GHD_conservation}, even in cases where the initial data is not differentiable (this follows from Theorem \ref{theo:potentialform} and Lemma \ref{lem:equivalenceweak}):
\begin{theorem}[Existence of solution to the weak GHD equation]\label{th1summary}
    Under the above assumptions, 
    if in addition $n_0(x,p)$ is continuous in $x$ (pointwise in $p$), then the weak form of the GHD equation \eqref{equ:intro_GHD_conservation} makes sense, and possesses a weak solution $\rho\ind{p}(t,x,p)$,
    \begin{align}
        \int\dd{t}\dd{x}\big(\partial_t\phi(t,x)\rho\ind{p}(t,x,p) + \partial_x\phi(t,x)v\upd{eff}(t,x,p)\rho\ind{p}(t,x,p)\big) &= 0,
    \end{align}
    for every $p\in \mathbb{R}$ and every Schwartz function $\phi \in \mathcal{S}(\mathbb{R}^2)$. In particular the solution $v\upd{eff}(t,x,p)$ to \eqref{equ:intro_GHD_conservation_velocity} exists and is unique for every $t,x,p$. The GHD solution $\rho\ind{p}(t,x,p)$ satisfies \eqref{equ:summary_bound_rho_t}, as well as the bounds \eqref{equ:summary_bound_n_t} (resp.~the bound \eqref{equ:summary_bound_n_t_2}) under \eqref{equ:summary_bound_n} (resp.~\eqref{equ:summary_bound_n_2}).
\end{theorem}
Note that a sligthly stronger version of the weak GHD equation -- its potential form \eqref{def:potential} -- leads to uniqueness as well, see Theorem \ref{theo:potentialform}.

If the initial data is differentiable this solution is indeed a strong solution to the GHD equation and is also unique (this follows from Theorem \ref{theo:differentiableform}): 
\begin{theorem}[Existence of a unique continuously differentiable solution to the GHD equation]
    Under the assumptions of Theorem \ref{th1summary}, if in addition
    $n_0(x,p)$ is continuously differentiable in $x$ (pointwise in $p$) and $\sup_{x,p} \abs{v(p)^2\partial_x n_0(x,p)} < \infty$, then there exists a unique solution $\rho\ind{p}(t,x,p)$ to the GHD equation that is continuously differentiable in $t$ and $x$ (pointwise in $p$) for all $t,x,p\in\mathbb R$.
    \label{thm:overview_strong_solution}
\end{theorem}
Note that in particular the solution will stay continuous for all times. This rules out the possibility of shock formation for initial conditions satisfying (\ref{equ:summary_bound_n}). Again, for quantum models like the Lieb-Liniger model, bound (\ref{equ:summary_bound_n}) is automatically satisfied and thus we show that shock formation is absent in those models for natural initial conditions.

If we additionally assume that $n_0(x,p)$ is smooth  (uniformly in $p$) and decays fast enough the solution will also stay smooth for all times (this follows from Theorem \ref{theo:differentiableform}):
\begin{theorem}[Smooth solution to the GHD equation, uniformly in $p$]
    Under the assumptions of Theorem \ref{th1summary}, if in addition $n_0(x,p)$ is smooth in $x$ (uniformly in $p$) and $\sup_{x,p} \abs{v(p)^a\partial_x^b n_0(x,p)} < \infty$ for all $a,b \geq 0, a,b\in \mathbb{Z}$, then $\rho\ind{p}(t,x,p)$ is smooth in $t$ and $x$ (uniformly in $p$) for all $t,x,p\in\mathbb R$.
\end{theorem}


\section{Derivation of the fixed-point problem}\label{sectoverview}

Before we dive into the proof we are going to outline the derivation of the fixed-point problem, which is at the heart of our proofs. Here we only focus only on the relevant concepts for this paper. We invite interested readers to have a look at our more physical companion paper~\cite{hübner2024newquadraturegeneralizedhydrodynamics}, where we discuss the results in more generality, providing more physical insights and connecting to previous works.

\subsection{The GHD equation}

The equation \eqref{equ:intro_GHD_conservation} (with \eqref{equ:intro_GHD_conservation_velocity}) is the conservation form of the GHD equation. Mathematically, this equation poses the problem of the existence and uniqueness of the solution to the integral equation \eqref{equ:intro_GHD_conservation_velocity}, before even being able to address the existence and uniqueness of the Cauchy problem itself.

It is known that the equation can be reformulated as a transport equation: One can show~\cite{PhysRevX.6.041065} (with technical assumptions that we will clarify below) that \eqref{equ:intro_GHD_conservation} with \eqref{equ:intro_GHD_conservation_velocity} imply
\begin{align}\label{ghdn}
    \partial_t n(t,x,p) + v\upd{eff}(t,x,p)\partial_x n(t,x,p) = 0.
\end{align}
This is the GHD equation in transport form. It indicates that $n(t,x,p)$ defined via \eqref{equ:summaryn} is a family (parametrised by $p$) of Riemann invariants for the hydrodynamic equation \eqref{equ:intro_GHD_conservation}. In this form we can directly see that $n(t,x,p)$ is transported along the trajectories of quasi-particles whose velocities are given by $v\upd{eff}(t,x,p)$. 

Using the function $n(t,x,p)$, one can obtain an expression for $v\upd{eff}(t,x,p)$ that is easier to analyse. One defines the dressing of a function $f(p)$ as the solution to the linear integral equation
\begin{align}
    f\upd{dr}(t,x,p) = f(p) + \int\tfrac{\dd{q}}{2\pi} \varphi(p,q) n(t,x,q) f\upd{dr}(t,x,q).\label{equ:GHD_dressing}
\end{align}
It is clear from this definition and from Eqs.~\eqref{equ:summaryn}, \eqref{equ:summaryrhos} that
\begin{align}\label{rhos1dr}
	2\pi\rho_{\rm s}(t,x,p) = 1\upd{dr}(t,x,p),\quad
    \rho\ind{p}(t,x,p) = n(t,x,p)\rho\ind{s}(t,x,p)
\end{align}
where $1(p)=1$ is the constant unit function. Most importantly, by a simple re-organisation of the terms, one can recast \eqref{equ:intro_GHD_conservation_velocity} into a dressing form,
\begin{align}\label{veffvdr}
	v\upd{eff}(t,x,p) = \frac{{v}\upd{dr}(t,x,p)}{1\upd{dr}(t,x,p)}.
\end{align}
With
\begin{equation}
    T(p,q) = \frac{\varphi(p,q)}{2\pi}
\end{equation}
seen as the integral kernel of the intergal operator $\vu{T}$, i.e.~$(\vu{T}f)(p) = \int\dd{q} T(p,q)f(q)$, the dressing takes the following form:
\begin{align}
    f\upd{dr}(t,x,p) = (1-\vu T n(t,x))^{-1} f(p)
\end{align}
where (by a slight abuse of notation) $n(t,x)$ is seen as the diagonal integral operator $(n(t,x) f)(p) = n(t,x,p)f(p)$. Thus, the existence and uniqueness of the effective velocity simply requires the invertibility of the linear operator $1-\vu T n(t,x)$.

From the thermodynamic Bethe ansatz~\cite{jsc_integrability,BenGHD}, the state density $\rho_{\rm s}(t,x,p)$ describes the density of available quantum states, while the occupation function $n(t,x,p)$ describes the fraction of the available quantum states that are occupied. In many quantum models, particles adhere to the fermionic statistics, and thus there is the restriction that each state is occupied at most once, implying $0 \leq n(t,x,p) \leq 1$~\cite{jsc_integrability}. In particular, $n(t,x,p) = 1$ means that each quantum state is filled once, while $n(t,x,p) = 0$ means that no state is filled. 

\subsection{Change of metric: trivialising the GHD equation}\label{ssectmetric}

One important result about the GHD equation is that it can be mapped, by an appropriate ``state-dependent change of metric", to the equation for the evolution of the phase-space density of free particles with velocity $v(p)$. This was first noted in~\cite{DOYON2018570}. Here, we assume that $\varphi(p,q)\geq 0$ and $\rho_{\rm p}(t,x,p)>0$.

One way to make this construction is as follows. Assume for simplicity that $\lim_{x\to-\infty}\rho_{\rm p}(t,x,p)= 0$ uniformly in $p$ and $t$, and quickly enough. Consider the new spatial variable
\begin{align}\label{xhatx}
	\hat x = x + \int_{-\infty}^x \dd{y} \int \dd{q} \varphi(p,q) \rho_{\rm p}(t,y,q).
\end{align}
This defines a function $\hat x = \hat X(t,x,p)$, which, by non-negativity of the integrand, is strictly monotonic in $x$ (for all $p$ and $t$); it can also be written as
\begin{equation}\label{Xhatrhos}
    \hat X(t,x,p) = x + \int_{-\infty}^x
    \dd{y} \Big(
    2\pi \rho_{\rm s}(t,y,p) -
    1\Big).
\end{equation}
Note in particular that, for $p,t$ fixed,
\begin{align}\label{metric}
	\dd \hat x = 2\pi\rho_{\rm s}(t,x,p) \dd x\qquad \mbox{($p,t$ fixed).}
\end{align}
Thus, this can be interpreted as a state-dependent change of metric. Further, its time dependence is fully determined as follows: for $p$ fixed in the asymptotic region $x\to-\infty$, 
\begin{equation}\label{equ:freecoordtime}
	\dd\hat x = \dd x + 0 \dd t \qquad 
	\mbox{($x\to -\infty$, $p$ fixed).}
\end{equation}
Thus this change of coordinate is trivial and time-independent in the asymptotic region where the density vanishes. We will refer to $x$ as the ``real" coordinate, and $\hat x$ as the ``free" coordinate, for reasons that we now explain.

By monotonicity, one also has its inverse $x = X(t,\hat x,p)$. One then defines the ``free density" $\hat \rho_{\rm p} (t,\hat x,p)$ as
\begin{align}\label{rhohatrho}
	\hat \rho_{\rm p} (t,\hat x,p) \dd \hat x = \rho_{\rm p} (t, x,p) \dd  x.
\end{align}
Explicitly $\hat \rho_{\rm p} (t,\hat x,p) = \rho_{\rm p} (t, X(t,\hat x,p),p) \,\partial X(t,\hat x,p)/\partial \hat x$, where the derivative is nonzero by strict monotonicity. Note that in terms of the occupation function \eqref{equ:summaryn}, this is\footnote{One can naturally define $\hat \rho_s(t,x,p) = 1/(2\pi)$ as in the new metric the space density is constant, and thus $\hat n(t,\hat x,p) = n(t,X(t,\hat x,p),p)$.} $\hat \rho_{\rm p}(t,\hat x,p) = n(t,X(t,\hat x,p),p)/(2\pi)$. From this, one has a time-independent nonlinear map: given $\rho_p(t,\cdot,\cdot)$, one evaluates $\hat X(t,\cdot,\cdot)$ from \eqref{xhatx}, inverts it at fixed $t,p$, for all $p$, to get $X(t,\cdot,\cdot)$, and evaluates $\hat \rho_{\rm p} (t,\hat x,p)$,
\begin{align}\label{map}
	\rho_{\rm p}(t,\cdot,\cdot) \mapsto \hat \rho_{\rm p}(t,\cdot,\cdot).
\end{align}
It is a simple matter to show that {\em if $\rho_{\rm p}(t,x,p)$ satisfies the GHD equation, then $\hat \rho_{\rm p}(t,\hat x,p)$ satisfies the free-particle phase-space evolution equation}
\begin{align}\label{freeevol}
	\partial_t \hat\rho_{\rm p}(t,\hat x,p) + v(p)\partial_{\hat x} \hat\rho_{\rm p}(t,\hat x,p) = 0.
\end{align}
Thus, the GHD equation is trivialised under the map \eqref{map}. The solution in the trivialised form is immediate,
\begin{align}\label{freesol}
	\hat\rho_{\rm p}(t,\hat x,p) = \hat\rho_{\rm p}(0,\hat x-v(p)t,p).
\end{align}
This trivialisation is the rational for referring to $\hat x$ and $\hat \rho\ind{p}$ as ``free".

The physical interpretation of $\hat\rho_{\rm p}(t,\hat x,p)$ (see \cite{BenGHD}) is that this is the density of particles in the asymptotic coordinates of the microscopic model, in the sense of scattering theory; these, by construction, evolve trivially. The map \eqref{map} is the map to asymptotic coordinates, which is implemented by a change of metric \eqref{metric} representing the accumulation of two-body scattering shifts incurred in going to the asymptotic space.

\subsection{Solution via height fields}\label{heightfields}
The important new insight of our work is to consider these equations using the height fields:
\begin{align}\label{eq:defKheight}
    N(t,x,p) &= 2\pi \int_{-\infty}^x\dd{y} \rho\ind{p}(t,x,p) & \hat{N}(t,\hat{x},p) &= 2\pi \int_{-\infty}^{\hat{x}}\dd{\hat{y}} \hat\rho_{\rm p}(t,\hat y,p).
\end{align}
Note that by definition they satisfy
\begin{align}\label{equ:N_equivalence}
    N(t,x,p) = \hat{N}(t,\hat{X}(t,x,p),p).
\end{align}
and furthermore from \eqref{freesol} we can infer:
\begin{align}\label{equ:N_freesol}
    \hat{N}(t,\hat{x},p) = \hat{N}_0(\hat{x}-v(p)t,p),
\end{align}
where the initial $\hat{N}_0(\hat{x},p) = \hat{N}(0,\hat{x},p)$ can be computed from the initial condition.

Combining \eqref{equ:N_equivalence} and \eqref{equ:N_freesol} with \eqref{Xhatrhos} we have:
\begin{equation}\label{fp2}
    \hat X(t,x,p) = x + \vu{T} \hat N_0(\hat X(t,x,p) - v(p)t,p) =: G_{t,x}[\hat X(t,x,\cdot)](p)
\end{equation}
or alternatively 
\begin{equation}\label{fp1}
    N(t,x,p) = \hat N_0(x - v(p) t + \vu{T}N(t,x,p),p) =: H_{t,x}[N(t,x,\cdot)](p).
\end{equation}
By a slight abuse of notation, here and below $\vu{T}$, when acting on a function of many variables, implicitly acts on the momentum variable; for instance $\vu{T}N(t,x,p) = \int \dd{q} T(p,q)N(t,x,q)$ and $\vu{T} \hat N_0(K(t,x,p),p) = \int \dd{q} T(p,q)\hat N_0(K(t,x,q),q)$.

\subsection{Non-decaying initial conditions}
So far we assumed that $\rho(t,x,p) \to 0$ as $x\to - \infty$ sufficiently quickly. However, our proof also applies to cases, like the partitioning protocol
\begin{align}
    \rho_0(x,p) = \begin{cases}
        \rho\ind{L}(p) & x<0\\
        \rho\ind{R}(p) & x>0,
    \end{cases}
\end{align}
where the density does not decay. In these cases the definitions \eqref{xhatx} and \eqref{eq:defKheight} are not well-defined. It is convenient to fix the ambiguities not by conditions at $x = -\infty$, but rather at $x=0$ (or any other point). Here we only give the final result (which can easily be checked to indeed solve the GHD equation). Define
\begin{equation}
    \hat X_0(x,p) = 2\pi \int_0^x \dd x'
    \rho_{\rm s}(0,x',p),\quad
    N(0,x,p) = 2\pi \int_0^x \dd x'
    \rho_{\rm p}(0,x',p).
    \label{equ:K0N0}
\end{equation}
Furthermore define
\begin{equation}\label{N0rho}
    \hat N_0(\hat x,p) =
    2\pi \int_0^{X_0(\hat x,p)} \dd{x} \rho_{\rm p}(0,x,p).
\end{equation}
and, changing variable,
\begin{align}
    \hat N_0(\hat x,p) &= \int_0^{\hat x}\dd{\hat y} n_0(X_0(\hat y,p),p).\label{equ:derivation_N0_def}
\end{align}
Here $X_0(\hat x,p)$ is the inverse function to $\hat X_0(x,p)$ in $x$. With these definitions one can solve either \eqref{fp2} or \eqref{fp1} to compute $N(t,x,p) = \hat{N}_0(\hat{X}(t,x,p)-v(p)t,p)$. From this one obtains the solution for the particle density $\rho_{\rm p}(t,x,p)$ by differentiating $\partial_x N(t,x,p) = \rho_{\rm p}(t,x,p)$, or for the occupation function $n(t,x,p)$  from the fact that it satisfies the transport equation,
\begin{equation}\label{solutionnK}
    n(t,x,p)
    = \hat n_0 (\hat X_0(t,x,p)-v(p)t,p),\quad
    \hat n_0(\hat x,p) = n_0(X_0(\hat x,p),p).
\end{equation}

\subsection{Contraction property}

The main argument of this paper, in order to show existence and uniqueness of the GHD equation, will be to show existence and uniqueness of the fixed-point problems. It will be sufficient to choose one of them -- a similar analysis holds for all.

Let us analyse the fixed-point problem \eqref{fp2}. It turns out that the functional
\begin{align}
    G_{t,x}[\hat{X}](p) &= x + \vb{T}\hat{N}_0(\hat{X}(p) - v(p)t,p)
\end{align}
is contracting under the supremum norm $\norm{f}_{\infty} = \sup_{p\in{\mathbb{R}}} \abs{f(p)}$ if assumptions \eqref{equ:summary_bound_n} or \eqref{equ:summary_bound_n_2} hold:
\begin{align}
    & \norm{G_{t,x}[\hat{X}] - G_{t,x}[\hat{X}']}_\infty
    = \sup_{p\in\mathbb R}\abs{\vu{T}\int_{\hat{X}'(p)}^{
\hat{X}(p)} \dd{\hat y}
    n_0(X_0(\hat y,p),p)}
    \nonumber \\
    &\quad \leq \sup_{p\in\mathbb R} \int\dd{q}\abs{T(p,q)}\Big(\sup_{ x\in\mathbb R} \abs{n_0( x,q)}\,\abs{\hat{X}(q) - \hat{X}'(q)}
    \Big)\nonumber  \\
    &\quad \leq \tfrac{1}{2\pi}\qty(\sup_{p\in\mathbb R} \int\dd{q}\abs{\varphi(p,q)}\sup_{ x\in\mathbb R} \abs{n_0( x,q)})\norm{\hat{X} - \hat{X}'}_\infty <\norm{\hat{X} - \hat{X}'}_\infty.
\end{align}
Thus, by the Banach fixed-point theorem, the solution $\hat{X}(t,x,\cdot)\in B(\mathbb R)$ to \eqref{fp2}, a bounded measurable function on $p \in \mathbb R$,  exists and is unique. Once the solution $\hat{X}(t,x,p)$ is found, one can re-construct the solution to the GHD equation, either in terms of the particle density $\rho_{\rm p}(t,x,p)$ or the occupation function $n(t,x,p)$.

\subsection{Discussion} \label{ssectheightfielddiscussion}

A close look at the informal derivations above show that three conditions have appeared. (1) The existence of the dressing operation $(1-\vu{T}n(t,x))^{-1}$ for all $t,x$. This is guaranteed in particular if $\norm{\vu{T}n(t,x)}\ind{op}<1$, where we recall that $n(t,x)$ is the diagonal operator acting as multiplication by the occupation function $n(t,x) f(p) = n(t,x,p)f(p)$ on $f\in B(\mathbb R)$. This requires for each $t,x$
\begin{equation}\label{condntxp}
    \sup_p \int\dd{q} \abs{T(p,q)} n(t,x,q) < 1.
\end{equation}
In principle there is no need to impose this for all times, since as we will see, it follows from requiring at $t=0$ only, which in turn follows from \eqref{equ:summary_bound_n} or \eqref{equ:summary_bound_n_2}. (2) Invertibility (in $x$) of $\hat{X}(t,x,p)$. This was argued by assuming $\varphi(p,q)\geq 0$ and $n(t,x,p)\geq 0$. But a weaker requirement is $\inf_{(t,x,p)\in\mathbb R^3} \rho_{\rm s}(t,x,p) > 0$. By \eqref{rhos1dr}, this is immediately guaranteed if the stronger bound \eqref{equ:summary_bound_n_2} holds, as indeed, then, $(1-\vu Tn(t,x))^{-1} 1(p) = 1 + \vu Tn(t,x)(1-\vu Tn(t,x))^{-1} 1(p)> 1 - \norm{\vu Tn(t,x)}\ind{op}(1-\norm{\vu Tn(t,x)}\ind{op})^{-1}>0$ uniformly on $t,x$. (3) The contraction property of the fixed-point map $G_{t,x}$, for which we impose \eqref{equ:summary_bound_n} or its stronger version \eqref{equ:summary_bound_n_2}

Naturally, we have not discussed differentiability properties, neither have we derived the GHD equations from the fixed-point problem. Our formal proofs below establish that the conditions on initial data are sufficient and fill these other gaps, putting the above discussion on a rigorous basis.

We emphasise that the fixed-point problem provides meaningful solutions to the GHD dynamics even in the case where the initial data $n(0,x,p)$ is not differentiable. In that case, the GHD equation is satisfied in the weak form, as expressed in Section \ref{sectsummary}.

As mentioned in Section \ref{sectsummary}, it is not difficult to extend the GHD equation to more general measure space on momenta, instead of $\dd p,\mathbb R$, see for instance the discussion in \cite{doyon2018exact,10.21468/SciPostPhys.6.4.049}; one may still use the Banach space of bounded measurable functions with the supremum norm. In some cases, depending on the applications of the GHD equation, it may be that a different norm must be used. For the Lieb-Liniger, sinh-Gordon and other quantum mechanical models, the bound $\norm{n(t,x,p)}_\infty < 1$ is natural from the physics and thus $\banachspace$ is a natural candidate. Indeed, for the repulsive Lieb-Liniger model we have $T(p,q) = \tfrac{1}{2\pi}\tfrac{2c}{c^2+p^2}$, where $c > 0$ is the coupling constant, and for the sinh-Gordon model $T(p,q) = \tfrac{1}{2\pi}\tfrac{2}{\cosh(p-q)}$ (where here $p,q$ are to be interpreted as relativisic rapidities). In both cases, $T(p,q)>0$ and satisfies $\norm{\vu{T}}\ind{op} = 1$.

A simple example for a model where $\norm{n(t,x,p)}_\infty$ is not necessarily bounded is the hard rods model, with $T(p,q) = -d$ where $d > 0$ \cite{Boldrighini1983}. There any $p$-integrable, positive $n(t,x,p)$ is physically meaningful. Note that in this case \eqref{fp2} implies that $\hat{X}(t,x,p)$ is independent of $p$ and therefore $G_{t,x}$ becomes a continuous and strictly decreasing function $\mathbb{R} \to \mathbb{R}$, thus always having a unique fixed point.

\section{Mathematical setup and dressing operation}
\label{sec:basics}

In this section we specify the mathematical setup, define some basic space of functions and their properties, and express the main assumptions within this language. We then establish important properties of the dressing operation and show how the initial conditions, that were expressed in Section \ref{sectsummary} in terms of the particle density $\rho\ind{p}(0,x,\theta)$, can be expressed solely in terms of the occupation function $n(0,x,\theta)=n_0(x,\theta)$.

\subsection{Notations and spaces of functions}\label{ssect:notations}

For any subset $S$ of $\mathbb R^D$ for some $D\geq 1$, we denote the corresponding Banach space of Lebesgue measurable, bounded real functions on $S$ as $B(S) = B(S\to\mathbb R)$. The non-negative subset of $B(S)$ is denoted $B_+(S) = \{f\in B(S):f(s)\geq 0 \;\forall\;s\}$. For any $\mu\in B_+(S)$, we denote $B(S,\mu) = B(S\to \mathbb R,\mu)$ as the space of real functions on $S$ with finite seminorm $\sup_{s\in S}\mu(s)\abs{f(s)}<\infty$. Note that if the support of $\mu$ is $S$, this is a norm and induces a Banach space; but we admit the general case. Hence the (semi)norms are
\begin{equation}
    \norm{f}_\infty^\mu := \sup_{s\in S}\mu(s)\abs{f(s)},\qquad
    \norm{f}_\infty := \sup_{s\in S}\abs{f(s)}.
\end{equation}
We note that if $\mu=n$ is bounded from above, then $B(S)\subset B(S,n)$, while if $\mu=\nu$ has a strictly positive lower bound, $B(S,\nu)\subset B(S)$. Also, in general, $f\in B(S,\abs{g}) \Leftrightarrow g\in B(S,\abs{f}) \Leftrightarrow fg\in B(S)$.

We consider the class of integral operators $\vu{A}:B(S)\to (S\to\mathbb R)$, with kernel
$A(s,s')$ such that $A(s,\cdot)\in L^1(S)$  for all $s\in S$, acting as $\vu{A}f(s) = \int \dd^D s' A(s,s')f(s')$. We define the operator norm as
\begin{equation}\label{equ:opnorm}
    \norm{\vu{A}}\ind{op}:=
    \sup_{s\in S}\int \dd^D{s'} 
    \abs{A(s,s')}.
\end{equation}
If $\norm{\vu{A}}\ind{op}<\infty$, then $\vu{A}:B(S)\to B(S)$.

Another type of operator is the multiplication operator by a bounded function: given $g\in B(S)$, this is $g: f\mapsto gf$ (we denote the operator using the same symbol as there should be no confusion). If $\mu\in B_+(S)$, then when explicitly stated, this will be seen as acting on the larger space $B(S,\mu)$, that is $\mu:B(S,\mu)\to B(S),\;f\mapsto \mu f$; then, $\mu$ preserves the (semi)norms.

Note that a bounded function $f(x,p)$, for $x\in S,\,p\in\mathbb R$, can be seen equivalently as a function $f\in B(S\to \banachspace)$, or $f\in B(S\times\mathbb R)$, and similarly with more independent variables, etc. We will use the various viewpoints interchangeably. Also, our convention is that when the last argument $p$ is omitted, the object is seen as an element of a space of functions of $p$ (similarly when the last two $s,p$ are omitted, as an element of a space of functions of $s,p$, etc.).

It will be important to discuss differentiability of functions. For $s\in \mathbb R^D$, we use the notation $\partial_s^a = \prod_{i=1}^D \partial_{s_i}^{a_i}$, and $\abs{a} = \sum_i a_i$, for $a = (a_i)_i$ with $a_i\in\mathbb N = \{0,1,2,\ldots\}$. We recall that for $S\subset \mathbb R^D$, a function $S\ni s\mapsto f(s)\in\mathbb R$ is $r$ times continuously differentiable on $S$, i.e.~$f\in C^r(S)$, if and only if all its derivatives $\partial_s^a f(s)$ exist and are continuous, for all $a$ such that $\abs{a}=0,1,\ldots,r$ and $s\in S$.

We will concentrate on two families of differentiability classes of functions of two variables $s\in S$ and $p\in \mathbb R$, with differentiability in $s$ being the focus. Again, this will in general be with respect to $\nu\in B_+(S)$, and in all our application, it will be sufficient to restrict to functions that are constant along $S$, and that have a strictly positive lower bound,
\begin{equation}\label{equ:muboundedbelow}
    \nu(s,p) = \nu(p),\quad
    \inf_{p\in \mathbb R} \nu(p)>0.
\end{equation}
Given $r\in\mathbb N$, and $S\subset \mathbb R^D$, we consider functions $f(s,p)$ that are $p$-pointwise $r$ times continuously differentiable in $s\in S$, with the additional property that their derivatives of orders $a$ up to the order $\abs{a}=r$ are uniformly bounded on $S\times \mathbb R$ with respect to the $\abs{a}$th power $\nu(p)^{\abs{a}}$ (below we use the notation $\nu^{\abs{a}}:s,p\mapsto \nu(p)^{\abs{a}}$):
\begin{equation}\label{equ:defCp}
    C^r\ind{p}(S,\nu) = \left\{\begin{aligned} f:&\ S\times \mathbb R \to \mathbb R,\\ &\ 
    f(\cdot,p)\in C^r(S) \;\forall\;p\in\mathbb R,\\ &\ 
    \partial_s^a f \in B(S\times \mathbb R,\nu^{\abs{a}})\;\forall\;a:\,1\leq \abs{a}\leq r
    \end{aligned}
    \right\}.
\end{equation}
We also consider functions $f(s,p)$ that are $p$-uniformly $r$ times continuously differentiable in $s\in S$, with again the condition on boundedness of their derivatives with respect to powers of $\nu$. This can equivalently be expressed as being $r$ times continuously differentiable as a function mapping to the Banach space $\banachspace$, with that condition on the derivatives:
\begin{equation}\label{equ:defCu}
    C^r\ind{u}(S,\nu) = 
    \left\{\begin{aligned} f: &\  S\times \mathbb R \to \mathbb R,\\ &\ 
    f\in C^r( S \to \banachspace),\\ &\ 
    \partial_s^a f \in B(S\times \mathbb R,\nu^{\abs{a}})\;\forall\;a:\,1\leq \abs{a}\leq r
    \end{aligned}
    \right\}.
\end{equation}
Note that in both cases, because of \eqref{equ:muboundedbelow}, one has in fact $\sup_{s\in S,p\in\mathbb R} (\nu(p)^b\, \abs{\partial_s^af(s,p)}) <\infty$, that is $\partial_s^a f \in B(S\times \mathbb R,\nu^{b})$, for all $0\leq b\leq \abs{a}$, $1\leq \abs{a}\leq r$.

When $\nu=1$, the last argument is again dropped:
\begin{equation}
    C^r\ind{u}(S) = C^r\ind{u}(S,1),\quad
    C^r\ind{p}(S) = C^r\ind{p}(S,1).
\end{equation}
These spaces restrict only the derivatives (of order at least 1) to be bounded; below we will combine with various conditions on boundedness of the function itself.

We note the (strict) inclusions $C^{r+1}\ind{p}(S,\nu)\subset C^{r}\ind{u}(S,\nu)\subset C^r\ind{p}(S,\nu)$ for all $r\geq 0$, and the equality $C^\infty\ind{p}(S,\nu) = C^\infty\ind{u}(S,\nu)$; the first inclusion is thanks to the bounds on the derivatives, as using integrations to express lower derivatives we see that they exist uniformly. Also, $C^0\ind{u}(S) = C^0(S\to \banachspace)$.

We also make the following observations. If $f\in C^r\ind{u}(S,\nu)$ for some $r\geq 0$ then $f\in B(U\times \mathbb R)$ for all compact subsets $U\subset S$. If $f\in C^r\ind{p}(S,\nu)$, then $\partial_s^a f(s,p)$ is $p$-uniformly continuous on $s\in S$ for all $\abs{a}< r$; this can be seen by integration. Therefore, $f\in C^r\ind{u}(S,\nu)$ if and only if $f\in C^r\ind{p}(S,\nu)$ and $\partial_s^af(s,p),\,\abs{a}=r$ are $p$-uniformly continuous for all $s\in S$. The pointwise and uniform families, in our definitions, are thus rather close to each other; this is because of the conditions on boundedness of the derivatives. 
Finally, and importantly, the spaces behave as follows under multiplication of the $s$ argument by $\nu(p)$:
\begin{equation}\label{equ:Crcomposition}
    f\in C^r\ind{p}(S,\nu) \quad \Rightarrow\quad \{(s,p)\mapsto f(\nu(p) s,p)\} \in C^r\ind{p}(S)
\end{equation}
and there are clearly more general relations of this type.

The spaces $B(S,n)$, $B_+(S,n)$ (for upper bounded non-negative $n$), and $C^r\ind{p}(S,\nu)$ and $C^r\ind{u}(S,\nu)$ (for non-negative $\nu$ with strictly positive lower bound), are the spaces we will use in order to express our conditions and results.

\subsection{Review of setup and assumptions}\label{ssect:reviewsetup}

Our setup and the Assumptions are stated in Subsection \ref{ssectassumptions}. We here re-state them in a way that is better adapted to our proof methods (without re-stating the conditions on Lebesgue measurability).

Two basic objects define the GHD equation: the scattering shift and the velocity function. These are  throughout assumed to be given.

Let us introduce the notation
\begin{equation}
    T(p,q) = \frac{\varphi(p,q)}{2\pi}.
\end{equation}
We assume given the scattering shift $T:\mathbb R^2\to\mathbb R$, with $T(p,\cdot)\in L^1(\mathbb R)$. Seeing this as an integral kernel, we have an integral operator $\vu{T}$ as defined around \eqref{equ:opnorm}. We further assume
\begin{equation}\label{equ:condTop}
    \norm{\vu{T}}\ind{op}
    < \infty.
\end{equation}
Therefore $\vu T:\banachspace\to\banachspace$; this is the bound \eqref{equ:supvarphi}. For any function $f\in B(\mathbb{R})$, we denote $\vu{T}f$ the integral operator with kernel $T(p,q)f(q)$.

Further, we assume given a real function
\begin{equation}
	\mathbb R\ni p\mapsto v(p)\in\mathbb R,
\end{equation}
the velocity function.

Another important function is the occupation function $n_0(x,p)$, which will be related to the initial condition. For it, we extract labeled assumptions, which are not assumed throughout and will be stated when needed.

First, we must consider the possibility \eqref{equ:summary_bound_n} of our Assumptions; that is, this is the condition that the scattering shift have a fixed sign, which we write as \eqref{assumptionS}:
\begin{equation}\label{assumptionS}
    \mbox{either}\quad T(p,q) \geq 0\ \forall\  p,q\in\mathbb R\quad \mbox{or}\quad T(p,q) \leq 0\ \forall\  p,q\in\mathbb R. \tag{S}
\end{equation}
Second, we define certain subsets of the set of non-negative bounded functions $B_+(S\times \mathbb R)$, determined by $\vu{T}$:
\begin{align}
    B\ind{T}(S\times \mathbb R) := \left\{n\in B_+(S\times \mathbb R): \norm{\vu{T} (\sup_{s\in S} n(s,\cdot))}\ind{op} < 1\right\}
\end{align}
and
\begin{align}
    B\ind{T}'(S\times \mathbb R) := \left\{n\in B_+(S\times \mathbb R): \norm{\vu{T} (\sup_{s\in S} n(s,\cdot))}\ind{op} < \frac12\right\}
	\subset B\ind{T}(S\times \mathbb R).
\end{align}
By convention, if the space $S=\emptyset$ is empty, then $\sup_{s\in S} n(s,\cdot)$ is replaced by $n$. Note in general the following inequalities, which we will use throughout:
\begin{equation}
    \norm{\vu{T}n(s,\cdot)}\ind{op}\leq \norm{\vu{T}n}\ind{op} \leq \norm{\vu{T}(\sup_{s'\in S} n(s',\cdot))}\ind{op}\leq \norm{\vu{T}}\ind{op}\norm{n}_\infty\quad \forall\;s\in S
\end{equation}
where $\vu{T}$ only acts on the momentum argument of $n$; the operator norm $\norm{\cdot}\ind{op}$ in the first, third and fourth expressions are for operators acting on $B(\mathbb R)$, while that in the second is for the operator $\vu{T}n$ acting on $B(S\times \mathbb R)$, i.e.~$\norm{\vu{T} n}\ind{op} = \sup_{s\in S,\,p\in \mathbb R}\int \dd{q} \abs{T(p,q)}n(s,q)$. Thus in particular,
\begin{equation}\label{equ:boundTop}
    n\in B\ind{T}(S\times \mathbb R) \Rightarrow
    \norm{\vu{T} n}\ind{op}<1,\quad
    n\in B\ind{T}'(S\times \mathbb R) \Rightarrow
    \norm{\vu{T} n}\ind{op}<\frac12.
\end{equation}

We set $n_0$ as follows. If Assumption \eqref{assumptionS} holds (resp.~it doesn't), then $n_0\in B\ind{T}(\mathbb R^2)$ (resp.~$n_0\in B\ind{T}'(\mathbb R^2)$); this is \eqref{equ:summary_bound_n} (resp.~\eqref{equ:summary_bound_n_2}). Explicitly, $n_0(x,p)\geq 0$ for all $x,p$ and
\begin{equation}
    \norm{\vu{T}(\sup_{x\in\mathbb R}n_0(x,\cdot))}\ind{op} < 1\qquad \Big(\mbox{resp.}\ 
    \norm{\vu{T}(\sup_{x\in\mathbb R}n_0(x,\cdot))}\ind{op} < \frac1{2}
    \Big).
\end{equation}
It will be convenient to introduce the positive subset of the open ball of radius that depends if the fixed-sign condition \eqref{assumptionS} holds or not:
\begin{equation}\label{equ:defBS}
    B\ind{S}(S\times \mathbb R) = \left\{\begin{array}{ll}
         B\ind{T}(S\times \mathbb R) &\mbox{(the fixed-sign condition \eqref{assumptionS} holds)}  \\
         B\ind{T}'(S\times \mathbb R) &\mbox{(otherwise).}
    \end{array}\right.
\end{equation}
Then we write as \eqref{assumptionA}:
\begin{equation}
    \label{assumptionA}
    \tag{A}
    n_0 \in B\ind{S}(\mathbb R^2),\quad \mbox{that is:}\quad
    \mbox{either \eqref{assumptionS} holds and }
    n_0\in B\ind{T}(\mathbb R^2),\ 
    \mbox{or }
    n_0\in B\ind{T}'(\mathbb R^2).
\end{equation}
By \eqref{equ:boundTop}, this implies (but is stronger than)
\begin{equation}
    \mbox{either \eqref{assumptionS} holds and }\norm{\vu{T}n_0}\ind{op} < 1, \quad \mbox{or }
    \norm{\vu{T}n_0}\ind{op} < \frac12,
\end{equation}
and also it implies $n_0(x)\in B\ind{S}(\mathbb R)$ for all $x\in\mathbb R$.

Second, another condition on $n_0$ is that with respect to the velocity function. In order to state it, it is convenient to see $v$ as a function $\mathbb R^2\to\mathbb R:(x,p)\to v(p)$ that is independent of the first variable $x$. Then the condition is
\begin{equation}
    \label{assumptionV}
    \tag{V}
    n_0 \in B(\mathbb R,\abs{v}) \quad\mbox{that is}\quad
    \sup_{(x,p)\in\mathbb R^2}\abs{v(p)}\,n_0(x,p)<\infty.
\end{equation}
This can be written in different equivalent ways, such as $v\in B(\mathbb R^2,n_0)$ and $v\in B(\mathbb R,\sup_{x\in\mathbb R}n_0(x,p))$. This implies (but is stronger then) $v\in B(\mathbb R,n_0(x))$ for all $x\in\mathbb R$.

Below we will use the following strict inequalities:
\begin{equation}
    0<R(\norm{\vu{T}n_0}\ind{op})<\infty,\quad 0<\norm{n_0}_\infty <\infty,\quad 0<\frac1{1-\norm{\vu{T}n_0}\ind{op}}<\infty
\end{equation}
where we define the following function, which depends on if the condition \eqref{assumptionS} holds or not:
\begin{equation}
    \label{equ:functionR}
    R(z) = 
    \left\{\begin{array}{ll}
    1-z & \mbox{(the condition (\ref{assumptionS}) holds)}\\[0.3cm] \displaystyle
    \frac{1-2z}{1-z} & \mbox{(otherwise).}
    \end{array}\right.
\end{equation}

The Assumptions in Section \ref{ssectassumptions} boil down to the existence of $\rho\ind{p}(0,x,t)$, related to $n_0(x,p)$ via Eqs.~\eqref{equ:summaryn} and \eqref{equ:summaryrhos}, such that \eqref{assumptionA}, \eqref{assumptionV} hold. We now develop the theory for the dressing operation. We will show that given any $n_0$ such that \eqref{assumptionA} holds, then $\rho\ind{p}(0,x,p)$, $\rho\ind{s}(0,x,p)$ are unique and satisfy $\rho\ind{p}(0,x,p)\geq 0$, $\rho\ind{s}(0,x,p)>0$. Thus assuming \eqref{assumptionA}, \eqref{assumptionV} is sufficient: the full space of such functions $n_0$ is available, and there is no need to additionally assume the existence of $\rho\ind{p}(0,x,t)$.

In Section \ref{ssectassumptions} we considered $n_0(x,t)$ as the initial occupation function. Our proof methods is based on the fixed-point problem \eqref{fp2}, \eqref{fp1}, which is determined by the function $\hat N_0(x,p)$, from which we will define an occupation function $n(t,x,p)$. We will construct $\hat N_0(x,p)$ in terms of $n_0(x,t)$, and only later show that indeed this fixes the initial condition $n(0,x,p)$. We therefore keep the notation separate, and refer to $n_0$ as the {\em seed occupation function} and to $\hat N_0$ as the {\em seed height function}.

\subsection{Properties of the dressing operation}

We recall that as per our convention, for $n\in \banachspaceplus$, $\vu{T}n: B(\mathbb R) \to \banachspace$ is defined as
\begin{align}
    \vu{T}n
    f(p) = \int\dd{q}T(p,q)n(q)f(q).
\end{align}
Clearly, for every $n\in\banachspaceT$, the operator $\vu{T}n$ is a bounded linear operator with $\norm{\vu{T}n}\ind{op}<1$, and thus $(1-\vu{T}n)^{-1}$ also is a bounded linear operator, with
\begin{equation}\label{equ:bounddressingnormal}
    \norm{(1-\vu{T}n)^{-1}}\ind{op}\leq
    \frac{1}{1-\norm{\vu{T}n}\ind{op}}.
\end{equation}
As per our comments above, we may also see $\vu{T}n: B(\mathbb R,n) \to \banachspace$ (note that in this viewpoint, $\vu{T}n$ has norm $\norm{\vu{T}}\ind{op}$). Thus, by interpreting as such the right-most factor in each term $(\vu{T}n)^a,\,a=1,2,\ldots$ in the Taylor series expansion of $(1-\vu{T}n)^{-1}$, the Taylor series still converges and we have extended $(1-\vu{T}n)^{-1}$ to an operator $B(\mathbb R,n)\to B(\mathbb R,n)$ (as $\banachspace\subset B(\mathbb R,n)$). This is what we refer to as the ``dressing operation" (this is a precise formulation of the well-known concept); here we provide a definition, along with simple statements that follow from it:
\begin{definition}[Dressing operation]
    The dressing operation with respect to an occupation function $n\in B_+(\mathbb R)$ such that $\norm{\vu{T}n}\ind{op}<1$ is
\begin{align}
    {\cdot}\upd{dr}\  :\   B(\mathbb R,n) &\to B(\mathbb R,n)  \nonumber\\
    f &\mapsto f\upd{dr}
    = \sum_{a=0}^\infty (\vu{T}n)^a f
    = f + (1-\vu{T}n)^{-1}\vu{T}n f
    \label{equ:def_dressing_extended}
\end{align}
with $f\upd{dr}-f\in\banachspace$. For any $f\in B(\mathbb R,n)$, this is the unique $f\upd{dr}\in B(\mathbb R,n)$ solving $f\upd{dr} = f +\vu{T}nf\upd{dr}$. In particular, the dressing maps $\banachspace\to\banachspace$ on which it acts as $(1-\vu{T}n)^{-1}$.
\label{def:dressing_general}
\end{definition}
The bound with respect to the norm on $B(\mathbb R,n)$ is immediate again from the Taylor series expansion,
\begin{equation}\label{equ:bounddressingn}
    \norm{f\upd{dr}}_\infty^n \leq \frac{\norm{f}_\infty^n}{1-\norm{\vu{T}n}\ind{op}},\qquad f\in B(\mathbb R,n).
\end{equation}
By \eqref{equ:bounddressingnormal}, the same form of inequality holds for $f,f\upd{dr}\in\banachspace$  with $\norm{\cdot}_\infty$ (which is a stronger inequality, valid for $f$ in that subspace). Further, we similarly have
\begin{equation}
\label{equ:bounddressingnff}
    \norm{f\upd{dr}-f}_\infty \leq \frac{\norm{\vu{T}}\ind{op}\norm{f}_\infty^n}{1-\norm{\vu{T}n}\ind{op}},\qquad f\in B(\mathbb R,n).
\end{equation}

It will be important for the dressing operation to behave well under variations of parameters on which $n$ depend: time and space. In particular, we will be looking for continuity and differentiability. We will consider these properties in two ways: pointwise in $p$ ($p$-pointwise), and uniformly in $p$ ($p$-uniformly).

\begin{lemma}[Differentiability class of the dressing] Given $S\subset \mathbb R^D$, let $n\in B_+(S\times \mathbb R)$ with $\norm{\vu{T}n}\ind{op}<1$ and $f \in B(S\times \mathbb R,n)$. Let $r\in\mathbb N$. (Case 1) If $n,\,nf\in C^r\ind{p}(S)$ then $f\upd{dr}-f\in C^r\ind{p}(S)$. (Case 2) If $n,\,nf\in C^r\ind{u}(S)$ then $f\upd{dr}-f\in C^r\ind{u}(S)$. Further, for $r\geq 1$, the derivatives are obtained by the usual rules, for instance with $D=1$:
\begin{align}\label{equ:derfdr}
        \partial_x f\upd{dr}
        -
        (\partial_x f)\upd{dr} 
        = \qty[\vu{T}\partial_x nf\upd{dr}]\upd{dr} = (1-\vu{T}n)^{-1}\vu{T}\partial_x nf\upd{dr}.
    \end{align}
\label{lem:dressing_differentiability}
\end{lemma}
\begin{proof}
We use \eqref{equ:def_dressing_extended} and it is sufficient to concentrate on the second term, which is thus $f\upd{dr}-f$. Assume pointwise (case 1), resp.~uniform (case 2), differentiability.

Let $g=\vu{T}nf$. Then we show that $g(s,p)$ is  (case 1) $p$-pointwise $r$ times continuously differentiable in $s$ with $\partial_s^a g\in B(S\times \mathbb R)$, $\abs{a}=0,1,\ldots,r$; resp.~(case 2) $p$-uniformly $r$ times continuously differentiable in $s$. As a consequence $g\in B(S\times \mathbb R) \cap C^r\ind{p}(S)$, resp.~$g\in B(S\times \mathbb R) \cap C^r\ind{u}(S)$, thus the same holds for $ng$, and therefore $g$ has the same properties as $f$. By induction, the same holds for $(\vu{T}n)^c f$ for all integers $c\geq 1$.

That $g\in B(S\times \mathbb R)$ follows from $\norm{\vu{T}nf}_\infty\leq \norm{\vu{T}}\ind{op}\norm{f}_\infty^n<\infty$. For pointwise differentiability (case 1): this follows from the Leibniz integral rule for exchanging derivatives and integrals (which follows from the dominated convergence theorem and mean value theorem), using finiteness of the measure on $p$, $\norm{\vu{T}}\ind{op}<\infty$, and the uniform bounds for the derivatives of the integrands $\abs{\partial_s^a (n(s,p)f(s,p))}\leq \norm{\partial_s^a(nf)}_\infty
\;\forall\;s\in S,p\in\mathbb R$. These bounds guarantee that $\partial_s^a g \in B(S\times\mathbb R)$ for all $\abs{a}=1,2,\ldots,r$. For uniform differentiability (case 2): this follows likewise from $\norm{\vu{T}}\ind{op}<\infty$ and $p$-uniform differentiability of  $n(s,p)f(s,p)$.

Now take $g=nf$ and denote $m = \max_{a:\abs{a}=0,1,\ldots,r}(\norm{\partial_s^a n}_\infty,\norm{\partial_s^a g}_\infty)<\infty$. By bounding each term that appears when applying the product rule for derivatives in $\partial_s^a ((\vu{T}n)^c\vu{T}g)$, for any $a:\abs{a}\leq r$, we find that for all $c>r$ we have $\norm{\partial_s^{a} ((\vu{T}n)^c\vu{T}g)}_\infty \leq (c+1)^{\abs{a}} m^{\abs{a}}\norm{\vu{T}n}\ind{op}^{c-\abs{a}}(\norm{\vu{T}}\ind{op}^{1+\abs{a}}\norm{g}_\infty + \norm{\vu{T}}\ind{op}^{\abs{a}}\norm{\vu{T}n}\ind{op})$.
Because $\norm{\vu{T}n}\ind{op} < 1$, this is summable on $c$, hence we use this bound and the Leibniz rule for exchanging sum and derivative, (case 1) pointwise in $p$, or (case 2) uniformly in $p$. Boundedness of the result $\partial_s^a (f\upd{dr}(s,p)-f(s,p)),\;\abs{a}=1,\ldots,r$ follows, and as the series obtained by applying the product rule for derivatives in $\partial_s^a ((\vu{T}n)^c\vu{T}g)$ is absolutely convergent by the above discussion, the terms can be re-organised, to give for instance \eqref{equ:derfdr}.
\end{proof}

For functions $f:S\to\banachspace$, the result could also have been obtained by using the implicit function theorem~\cite{Abraham1988}: one constructs $F:\mathbb{R} \times \banachspace \times \banachspace \to \banachspace$:
\begin{align}
    F_x(f,g) &= f - (1-\vu{T}n(x))g
\end{align}
and uses $n\in C^r(S \to \banachspace)$ and the fact that $F_x(f,g_x(f)) = 0$ implies $g_x(f)(p) = f\upd{dr}(x,p)$. However, as we will the more general case of the above lemma, and as the proof method above works in essentially the same way for the pointwise and uniform cases, both were included.

\begin{lemma}[Differentiability of the dressing in momentum]
    Let $n\in B_+(\mathbb R)$ with $\norm{\vu{T}n}\ind{op}<1$, and $f\in B(\mathbb R,n)$. Let $I\subset \mathbb R$ be a finite interval. If $f(p)$, and ($q$-pointwise) $T(p,q)$, are continuous and $r$ times continuously differentiable in $p\in I$, and if $\norm{\sup_{p\in I} \abs{\partial^{b}_p T(p,\cdot)}}_1 < \infty$ for all $b = 0,1,2,\ldots,r$, then $f\upd{dr}(p)$ is continuous and $r$ times continuously differentiable in $p\in I$.
\end{lemma}
\begin{proof}
We use the bounded convergence theorem (for continuity) and the Leibniz integral rule (for continuous differentiability). Define $m^b(q) = \sup_{p\in I} \abs{\partial^{b}_p T(p,q)}$, which is integrable by the condition of the lemma. We bound
\begin{equation}
	\abs{\partial_q^bT(p,q)n(q)f\upd{dr}(q)} \leq m^b(q) \norm{nf\upd{dr}}_\infty \quad\forall\;p\in I,\,q\in\mathbb R.
\end{equation}
As $\norm{nf\upd{dr}}_\infty <\infty$ (by Definition \ref{def:dressing_general}), the function $m^b(q) \norm{nf\upd{dr}}_\infty$ is integrable on $q\in\mathbb R$. The lemma follows by writing
\begin{align}
    f\upd{dr}(p) &= f(p)+ \int\dd{q} T(p,q) n(q) f\upd{dr}(q).
\end{align}
\end{proof}


Finally, as is clear from our discussion in Section \ref{sectoverview}, one very important quantity in GHD is $1\upd{dr} = (1-\vu{T}n)^{-1} 1$: the dressing of the constant function $1 : p\mapsto 1$ (with a slight abuse of notation) with respect to $n$.  We will need to have precise bounds for this quantity. In particular, the lower bound is shown to be strictly positive either under the stronger condition \eqref{assumptionS} for the interaction kernel, but with the weaker condition $\norm{\vu{T}n}\ind{op}<1$ for the occupation function $n\in B(\mathbb R)$, or under no additional condition on the interaction kernel, but with the stronger condition $\norm{\vu{T}n}\ind{op}<1/2$. Note that these are consequences, on $n= n_0(x,\cdot)$ (for any fixed $x\in\mathbb R$), of the two cases assumed in Assumption \eqref{assumptionA} for $n_0$.
\begin{lemma}[Bounds on $1\upd{dr}$]
    Let $n\in B_+(\mathbb R)$, and assume that: either \eqref{assumptionS} holds and $\norm{\vu{T}n}\ind{op} < 1$, or $\norm{\vu{T}n}\ind{op} < 1/2$. Then, for all $p\in\mathbb R$:
    \begin{align}
        0 < R(\norm{\vu{T}n}\ind{op}) \leq 1\upd{dr}(p) \leq \tfrac{1}{1-\norm{\vu{T}n}\ind{op}} < \infty
    \end{align}
    where we make use of the function \eqref{equ:functionR}.
\label{lem:dressing_boundedness_of_1dr}
\end{lemma}
\begin{proof}
The upper bound follows from the operator norm \eqref{equ:bounddressingnormal}. If $n\in\banachspaceT$, the lower bound follows from rewriting the series as (here and below we keep the independent variable $p$ implicit):
\begin{align}
    1\upd{dr} &= \sum_{a=0}^\infty (\vu{T}n)^a 1 = \sum_{a=0}^\infty (\vu{T}n)^{2a} (1+\vu{T}n1) \geq 1+\vu{T}n1 \geq 1-\norm{\vu{T}n}\ind{op}
\end{align}
where we used that $\vu{T}n$ either has a non-positive or non-negative kernel and thus $(\vu{T}n)^{2a}$ maps non-negative functions to non-negative functions. If $n\in B'\ind{T}(\mathbb R)$, it follows from writing
\begin{equation}
	\abs{(1-\vu{T}n)^{-1} 1}= \abs{1 + \vu{T}n(1-\vu{T}n)^{-1}1}
	\geq 1 - \abs{\vu{T}n(1-\vu{T}n)^{-1}1}
\end{equation}
and
\begin{equation}
	\abs{\vu{T}n(1-\vu{T}n)^{-1}1} \leq
	\frac{\norm{\vu{T}n}\ind{op}}{
	1-\norm{\vu{T}n}\ind{op}
	}< 1
\end{equation}
where the last inequality follows from $\norm{\vu{T}n}\ind{op}<1/2$.
\end{proof}

\subsection{Initial conditions for particle and state densities}\label{ssect:initialcondition}

We now show that the existence of particle and state densities related to $n_0(x,p)$ as per Eqs.~\eqref{equ:summaryn}, \eqref{equ:summaryrhos} (where we take $t=0$ and replace $n(0,x,p)$ by the seed occupation function $n_0(x,p)$), and their non-negativity / positivity, follow from \eqref{equ:summary_bound_n}, \eqref{equ:summary_bound_n_2}. As mentioned, this implies that the Assumptions can be stated simply as \eqref{assumptionA} and \eqref{assumptionV}
without reference to the particle density.

\begin{lemma}\label{lem:initialdensities}
    Given \eqref{assumptionA}, there are unique $\rho\ind{p}(x,p)$ and $\rho\ind{s}(x,p)$ satisfying
    \begin{align}\label{equ:rhopproof}
    n_0(x,p) &= \frac{\rho\ind{p}(x,p)}{\rho\ind{s}(x,q)}\\ \label{equ:rhosproof}
    \rho\ind{s}(x,p) &= \frac1{2\pi}
    +\int_{\mathbb R}\dd{q} T(p,q)\rho\ind{p}(x,q).
    \end{align}
    These satisfy $\rho\ind{p}(x,p)\geq 0$ and $\rho\ind{s}(x,p)>0$ for all $(x,p)\in\mathbb R^2$.
\end{lemma}
\begin{proof}
The combination of Eqs.~\eqref{equ:rhopproof} and \eqref{equ:rhosproof} imply that 
\begin{equation}\label{equ:rhos1dr}
    2\pi\rho\ind{s}(x,p) = 1\upd{dr}_0(x,p),
\end{equation}
where $1\upd{dr}_0(x,p) = (1-\vu{T}n_0(x))^{-1}1(p)$ is the dressing, Definition \ref{def:dressing_general}, of the constant function $1: p\mapsto 1$ with respect to $n_0(x): p\mapsto n_0(x,p)$. The dressing operation is well defined thanks to the condition on $n_0$. Thus $\rho\ind{s}(x,p)$ is unique, and by Lemma \ref{lem:dressing_boundedness_of_1dr}, it is positive $\rho\ind{s}(x,p)>0$. By \eqref{equ:rhopproof} $\rho\ind{p}(x,p)$ is then also unique and non-negative.
\end{proof}

\section{The fixed-point problem} \label{sec:fixedpoint}

We now study the properties of the fixed-point problem \eqref{fp2} discussed in Section \ref{heightfields}. Two aspects are to be investigated: (1) the seed height function, that is the function $\hat N_0(x,p)$ that determines the fixed-point problem, which we will construct in terms of the seed occupation function $n_0$; and (2) the fixed point itself, from which the time-dependent occupation function $n(t,x,p)$ will be defined. We will in particular show that $n(0,x,p) = n_0(x,p)$, thus guaranteeing that the full set of initial conditions satisfying the properties of $n_0$ is available.

{\bf Throughout this section, we assume  \eqref{assumptionA} and \eqref{assumptionV} (see Section \ref{sec:basics}), without further stating it.}

\subsection{The seed height function}\label{ssectseed}

In this subsection, we construct $\hat N_0(x,p)$ and the related $\hat X_0(x,p)$ in terms of $n_0(x,p)$ (given in Assumption \eqref{assumptionA}). In particular, $\hat X_0$ will be shown, in subsection \ref{ssectinitcond}, to correspond to the initial condition at $t=0$ for the fixed-point problem.

With \eqref{equ:rhos1dr} and the change of metric \eqref{metric} to the ``free" coordinate $\hat x$, we define $\hat X_0(x,p)$ as:
\begin{align}\label{equ:K0definition}
    \hat X_0(x,p) = \int_0^x\dd{y} 1\upd{dr}_0(y,p)
\end{align}
where we recall that $1\upd{dr}_0(x,p) = (1-\vu{T}n_0(x))^{-1}1(p)$ is the dressing of the constant function $1$ with respect to $n_0(x): p\mapsto n_0(x,p)$. In accordance with our informal discussion, we will denote the inverse, the ``real" coordinate, as
\begin{equation}
	X_0(\cdot,p) = \hat X_0^{-1}(\cdot,p).
\end{equation}
The following Lemma establishes basic properties of these functions.
\begin{lemma}[Basic properties of $\hat X_0,\,X_0$]\label{lem:K0basic}
For every $x>y$ and $p\in\mathbb R$,
    \begin{equation}\label{equ:K0lip}
        R(\norm{\vu{T}n}\ind{op}) \leq \frac{\hat X_0(x,p)-\hat X_0(y,p)}{x-y}\leq
        \frac{1}{1-\norm{\vu{T}n}\ind{op}}
    \end{equation}
    and for every $\hat x>\hat y$,
    \begin{equation}\label{equ:K0invlip}
        1-\norm{\vu{T}n}\ind{op}  \leq \frac{X_0(\hat x,p)-X_0(\hat y,p)}{\hat x-\hat y}\leq
        \frac1{R(\norm{\vu{T}n}\ind{op})}
    \end{equation}
    where the function $R(z)$ is defined in \eqref{equ:functionR}.
Hence $\hat X_0(\cdot,p)$ is strictly increasing, invertible and a bijection of $\mathbb R$ for every $p$, $\mathbb R\ni x\mapsto \hat X_0(x)\in\banachspace$ is bi-Lipschitz, and the inverse $X_0(\cdot,p)$ has the same properties. 
\end{lemma}
\begin{proof}
Inequalities \eqref{equ:K0lip} follow immediately from \eqref{equ:K0definition}, Lemma \ref{lem:dressing_boundedness_of_1dr} and Assumption \eqref{assumptionA}, and \eqref{equ:K0invlip} then follows from \eqref{equ:K0lip}.
\end{proof}
Thanks to uniform boundedness of $1_0\upd{dr}(x,p)$, Lemma \ref{lem:dressing_boundedness_of_1dr}, $\hat X_0(x,p)$ is, for every $p$, a.e.~differentiable in $x$. Further, by the uniform lower bound in Lemma \ref{lem:dressing_boundedness_of_1dr}, the derivative is uniformly strictly positive. Thus, by the inverse function theorem and bijectivity of $X_0(\cdot,p)$, for every $p\in \mathbb R$ and a.e.~$\hat x\in\mathbb R$, the derivative of $X_0(\cdot,p)$ exists and is $\partial_{\hat x}X_0(\hat x,p) = 1/1_0\upd{dr}(X_0(\hat x,p),p)$. By continuity, we therefore have
\begin{equation}\label{equ:X0def}
	X_0(\hat x,p) = \int_0^{\hat x}\dd \hat y
	\frac1{1_0\upd{dr}(X_0(\hat y,p),p)}.
\end{equation}

The language of measures clarify the above, and puts on firm ground the change-of-metric picture that we discussed informally. We note that $\dd{\hat X_0(x)}$ is a $\banachspace$-valued measure that is absolutely continuous with respect to the Lebesgue measure. Thus we may evaluate its Radon-Nikodym derivative; likewise for $\dd{X_0(\hat x)}$. As these are uniquely defined a.e., and it is convenient to define them {\em everywhere} as
\begin{equation}\label{equ:measureXhat0}
	\frac{\dd{\hat X_0(x)}}{\dd{x}}
	= 1_0\upd{dr}(x),\qquad
	\frac{\dd{X_0(\hat x)}}{\dd{\hat x}}
	= \frac1{1_0\upd{dr}(X_0(\hat x,\cdot),\cdot)}.
\end{equation}

We can now establish differentiability properties for these change of coordinates. Recall that $C^r\ind{p}(\mathbb R)$ is the space of $p$-pointwise $r$ times differentiable functions of $x\in\mathbb R$ with boundedness conditions on the derivatives, and similarly $C^r\ind{u}(\mathbb R)$ for uniform differentiability (see \eqref{equ:defCp}, \eqref{equ:defCu}).
\begin{lemma}[Differentiability of $\hat X_0,\,X_0$]\label{lem:k0_uniformly}
    Let $r\geq 0$ be an integer. If $n_0\in C^r\ind{p}(\mathbb R)$, then $\hat X_0,\,X_0\in C^{r+1}\ind{p}(\mathbb R)$. If $n_0\in C^r\ind{u}(\mathbb R)$, then $\hat X_0,\,X_0\in C^{r+1}\ind{u}(\mathbb R)$.
\end{lemma}
\begin{proof}
From Lemma \ref{lem:dressing_differentiability} with $f(x,p)=1$ and $D=1$ we find that $1_0\upd{dr}\in C\ind{p}^{r}(\mathbb R)$, resp.~$1_0\upd{dr}\in C\ind{u}^{r}(\mathbb R)$ in the pointwise, uniform case respectively. Then by the definition \eqref{equ:K0definition}, $\hat X_0$ is $p$-pointwise, resp.~$p$-uniformly, $r+1$ times continuously differentiable in $x$. As $1_0\upd{dr}\in C\ind{p}^{r}(\mathbb R)$, and considering \eqref{equ:defCp}, \eqref{equ:defCu}, all derivatives $\partial_x^a \hat X_0(x,p)$ are uniformly bounded on $(x,p)\in \mathbb R^2$ for all $2\leq a\leq r+1$. Further, by the upper bound in Lemma \ref{lem:dressing_boundedness_of_1dr}, so is $\partial_x \hat X_0(x,p)$. Hence, $\hat X_0\in C\ind{p}^{r+1}(\mathbb R)$, resp.~$\hat X_0\in C\ind{u}^{r+1}(\mathbb R)$, in the pointwise, uniform case respectively.

As $\hat X_0(x,p)$ is $p$-pointwise $r+1$ times continuously differentiable in $x$, and as its derivative $\partial_x \hat X_0(x,p) = 1_0\upd{dr}(x,p)$ is uniformly strictly positive on $x\in \mathbb R$ for every $p$ by the lower bound in Lemma \ref{lem:dressing_boundedness_of_1dr}, then, by the inverse function theorem, $ X_0(\hat x,p)$ is $p$-pointwise $r+1$ times continuously differentiable in $\hat x\in \hat X_0(\mathbb R,p)=\mathbb R$. Here the equality $\hat X_0(\mathbb R,p)=\mathbb R$ follows from Lemma \ref{lem:K0basic}.

In order to establish the stronger statements that $X_0\in C\ind{p}^{r+1}(\mathbb R)$, resp.~$X_0\in C\ind{u}^{r+1}(\mathbb R)$, we need to analyse the derivatives themselves. These can be obtained by using \eqref{equ:X0def} and their derivatives w.r.t.~$\hat x$. We may proceed by induction, omitting the details of the exact expressions of such derivatives. Assume that $X_0\in C\ind{p}^{a}(\mathbb R)$, resp.~$X_0\in C\ind{u}^{a}(\mathbb R)$, for some $a\geq 0$. By Lemma \ref{lem:K0basic} and the result of the previous paragraph, this is true for $a=0$. Then, making use of uniform strict positivity of $1_0\upd{dr}(x,p)$, and of $1_0\upd{dr}\in C\ind{p}^{r}(\mathbb R)$, resp.~$1_0\upd{dr}\in C\ind{u}^{r}(\mathbb R)$, and evaluating the derivatives starting from the right-hand side of \eqref{equ:measureXhat0} (at the step $a=0$), we find that if $a\leq r$, then $X_0(\cdot,p)$ is $p$-pointwise, resp.~$p$-uniformly, $a+1$ times continuously differentiable, and its derivative $\partial_{\hat x}^{a+1} X_0(\hat x,p)$ is uniformly bounded on $x,p\in\mathbb R$. Thus, by induction, $X_0\in C\ind{p}^{r+1}(\mathbb R)$, resp.~$X_0\in C\ind{u}^{r+1}(\mathbb R)$.
\end{proof}

Finally, we define the seed height function $\hat N_0(\hat x,p)$ and establish its basic properties. It was defined (informally) in \eqref{equ:derivation_N0_def}, which we repeat here:
\begin{align}
    \hat N_0(\hat x,p) &= \int_0^{\hat x}\dd{\hat y} n_0(X_0(\hat y,p),p).\label{equ:N0_def}
\end{align}
\begin{lemma}[Properties of $
\hat N_0$]\label{lem:N0_uniformly}
For every $\hat x>\hat y$ and $p\in\mathbb R$,
    \begin{equation}\label{equ:N0variation}
        0\leq \frac{\hat N_0(\hat x,p)-\hat N_0(\hat y,p)}{\hat x-\hat y}
        \leq
        \norm{n_0}_\infty.
    \end{equation}
Therefore $\mathbb R\ni \hat x\mapsto \hat N_0(\hat x)\in\banachspace$ is Lipschitz and non-decreasing. Further, if $n_0\in C^r\ind{p}(\mathbb R)$ (resp.~$C^{r}\ind{u}(\mathbb R)$) for some $r\geq 0$, then $\hat N_0\in C^{r+1}\ind{p}(\mathbb R)$ (resp.~$C^{r+1}\ind{u}(\mathbb R)$).
\end{lemma}
\begin{proof}
Inequalities \eqref{equ:N0variation} follow from \eqref{equ:N0_def}, and the Lipschitz and non-decreasing properties follow from these inequalities.

For the last statement, its condition along with Lemma \ref{lem:k0_uniformly} we find that $\{(\hat y,p)\mapsto n_0(X_0(\hat y,p),p)\}$ is  an element of $C\ind{p}^{r}(\mathbb R)$, resp.~$C\ind{u}^{r}(\mathbb R)$. Then by the definition \eqref{equ:N0_def}, $\hat N_0$ is $p$-pointwise, resp.~$p$-uniformly, $r+1$ times continuously differentiable in $x$, and all derivatives $\partial_x^a \hat N_0(\hat x,p)$ are uniformly bounded on $(x,p)\in \mathbb R^2$ for all $2\leq a\leq r+1$. Uniform boundedness of $n_0(\cdot,\cdot)$ (Assumption \ref{assumptionA}) then guarantees that $\hat N_0\in C\ind{p}^{r}(\mathbb R)$, resp.~$\hat N_0\in C\ind{u}^{r}(\mathbb R)$, in the pointwise, uniform case respectively.
\end{proof}
The above Lemmas imply that for every finite interval $I\subset\mathbb R$, we have $\hat X_0,X_0,\hat N_0\in B(I\times \mathbb R)$, and in particular
    \begin{equation}\label{equ:N0bounded}
\norm{\hat N_0(\hat x)}_\infty
        \leq \abs{\hat x}\,\norm{n_0}_\infty.
    \end{equation}

\subsection{Fixed-point problem and occupation function in space-time}\label{ssectfixedpoint}

In this subsection, $\hat X_0,\, X_0$ and $\hat N_0$ are defined as in Subsection \ref{ssectseed}.

We are looking to solve the fixed-point problem \eqref{fp2}, which we recall here
\begin{equation}\label{fp2proof}
    \hat X(t,x,p) = x + \vu{T} \hat N_0(\hat X(t,x,p) - v(p)t,p),
\end{equation}
and study the properties of its solution. In the next subsection we will also study the properties of the height function defined as
\begin{equation}\label{equ:defN}
    N(t,x,p) = \hat N_0(\hat X(t,x,p)-v(p)t,p),
\end{equation}
from which $\hat X$ can be recovered as
\begin{equation}\label{XhatN}
    \hat X(t,x,p) = x + \vu{T}N(t,x,p).
\end{equation}

Putting \eqref{XhatN} into \eqref{equ:defN} one obtains a fixed-point problem for $N$, Eq.~\eqref{fp1}. Here we will not be studying this fixed-point problem; although a similar, but slightly more involved, analysis could be done.

Note that the fixed-point problem \eqref{fp2proof} involves the velocity function $v(p)$, which, as mentioned, satisfies \eqref{assumptionV}. Recall that $\vu{T}$ acts on the momentum variable when it is explicit, for instance $\vu{T}\hat N_0(\hat X(p)-v(p)t,p) = \int \dd{q}T(p,q)\hat N_0(\hat X(q)-v(q)t,q)$.

We recall the physical interpretation of the fixed-point problem: \eqref{fp2proof} involves a linear-in-time shift of $\hat X$ in the argument of $\hat N_0$, representing the free, linear evolution of the free-coordinate height field, $\hat N(t,\hat x,p) = \hat N_0(\hat x-v(p) t,p)$; and simultaneously identifies $\hat X$ with a modification of the real coordinate $x$, representing how the free coordinate is the real coordinate plus a term that represents the interaction with all particles, from \eqref{metric} in our informal discussion.

%
%

We will show in this subsection that \eqref{fp2proof} defines $\hat X$ uniquely. We will also show that the unique solution to the fixed-point equation \eqref{fp2proof} at $t=0$ is $\hat X(0,x,p) = \hat X_0(x,p)$, and that the picture of a change of metric $x\mapsto \hat x$ is rigorously expressed as a change of measure, with $\dd{\hat x}$ absolutely continuous with respect to $\dd{x}$ at all times $t$. Then, in the next section, we will show how this gives rise to the GHD equation.

\begin{theorem}[Existence and uniqueness of solutions] For every $x,t\in\mathbb R$, Eq.~\eqref{fp2proof} has a unique solution $\hat X(t,x)\in\banachspace$.
\label{theo:fixedpoint}
\end{theorem}
\begin{proof}
For every $t,x\in\mathbb R$, let us consider the nonlinear map $G_{t,x}$ acting on $\banachspace$ defined as
\begin{equation}\label{equ:mapG}
	G_{t,x}[f](p) = x + \vu{T} \hat N_0(f(p) - v(p)t,p).
\end{equation}
We first show that its image lies within $\banachspace$, that is $G_{t,x}:\banachspace\to\banachspace$. For this purpose, we note that $G_{0,0}:\banachspace\to\banachspace$ because $\vu{T}$ is bounded and by \eqref{equ:N0bounded}. More precisely,
\begin{eqnarray}
    \norm{G_{0,0}[f]}_\infty &=&
    \sup_{p\in\mathbb R}\abs{\vu{T}\int_0^{f(p)} \dd{\hat x} n_0(X_0(\hat x,p),p)}\nonumber\\
    &\leq&\sup_{p\in\mathbb R} \int\dd{q}
    \abs{T(p,q)} \sup_{\hat x\in\mathbb R} n_0(X_0(\hat x,q),q)
    \abs{f(q)}\nonumber\\
    &\leq&\sup_{p\in\mathbb R} \int\dd{q}
    \abs{T(p,q)} \sup_{x\in\mathbb R} n_0(x,q)
    \norm{f}_\infty\nonumber\\
    &=& \norm{\vu{T}(\sup_{x\in\mathbb R}n_0(x,\cdot))}\ind{op}
    \norm{f}_\infty
    \label{equ:boundH00}
\end{eqnarray}
Then, for every $x,t,x',t'\in\mathbb R$,
\begin{align}
        \abs{G_{t,x}[f](p)-G_{t',x'}[f](p)}
        &\leq \abs{x-x'} + \abs{\vu{T}\int_{f(p)-v(p)t'}^{f(p)-v(p)t}\dd{\hat x}n_0(X_0(\hat x,p),p)}\nonumber\\
        &\leq \abs{x-x'} + \norm{\vu{T}}\ind{op}\sup_{(p,y)\in\mathbb R^2} \abs{v(p)(t-t')}        n_0(y,p)\nonumber\\
        &\leq \abs{x-x'}
        + \abs{t-t'}\,\norm{\vu{T}}\ind{op}\norm{v}_\infty^{n_0}.
        \label{equ:GtxG00}
\end{align}
Setting $x'=t'=0$, we find the desired finite bound for every $x,t\in\mathbb R$.

Second, we show that the map is a contraction, with contraction rate bounded by $\norm{\vu{T}(\sup_{x\in\mathbb R}n(x,\cdot))}\ind{op}<1$ (again the inequality follows from Assumption \eqref{assumptionA}), that is
\begin{equation}\label{equ:GHcontraction}
	\norm{G_{t,x}[f]-G_{t,x}[f']}_\infty \leq \norm{\vu{T}(\sup_{x\in\mathbb R}n_0(x,\cdot))}\ind{op}
        \norm{f-f'}_\infty.
\end{equation}
Indeed we have
\begin{eqnarray}
	\abs{G_{t,x}[f](p)-G_{t,x}[f'](p)}&\leq& \abs{\vu{T}\int_{f'(p)-v(p)t}^{f(p)-v(p)t}\dd{\hat x}n_0(X_0(\hat x,p),p)}\nonumber\\
    &\leq&
    \int_{\mathbb R} \dd{q}
    \abs{T(p,q)}
    \sup_{x\in\mathbb R}n_0(x,q)
    \abs{f(p)-f'(p)}
\end{eqnarray}
from which \eqref{equ:GHcontraction} follows.

Finally, the Banach fixed-point theorem implies the statement of the present theorem.
\end{proof}

With $\hat X$ given, we can then define the time-dependent occupation function
\begin{equation}\label{equ:defnt}
    n(t,x,p) = \hat n_0(\hat X(t,x,p)-v(p)t,p),\quad\mbox{with}\ 
    \hat n_0(\hat x,p) = n_0(X_0(\hat x,p),p).
\end{equation}
It is easy to see that for every $t,x,p$ it is bounded as (recall that $\norm{n}_\infty$ is with respect to $B(\mathbb R^3)$ because $n$ is a function on $\mathbb R^3$)
\begin{equation}\label{equ:ntbound1}
	\norm{n}_\infty \leq \norm{n_0}_\infty.
\end{equation}
Further, it is simple to show that $n(t)$ is in the same class as $n_0$ for all $t$ (see Assumption \eqref{assumptionA}), and in fact $n$, as a function of two variables in addition to the momentum, is in the corresponding class. Indeed,
\begin{equation}
    \norm{\vu{T}
    (\sup_{x\in\mathbb R}
    n(t,x,\cdot))}\ind{op}
    \leq
    \norm{\vu{T}
    (\sup_{(s,x)\in\mathbb R^2}
    n(s,x,\cdot))}\ind{op}
    \leq
    \norm{ \vu{T}
    (\sup_{x\in\mathbb R}
    n_0(x,\cdot))}\ind{op}
\end{equation}
which follows from
\begin{equation}
    \sup_{x\in\mathbb R}
    n(t,x,p)\leq
    \sup_{(s,x)\in\mathbb R^2}
    n(s,x,p)\leq
    \sup_{\hat x\in\mathbb R}
    \hat n_0(\hat x,p)
    \leq
    \sup_{x\in\mathbb R}
    n_0(x,p)\quad \forall\;t,p\in\mathbb R.
\end{equation}
That is, we have
\begin{equation}
    n\in B\ind{S}(\mathbb R^3),\quad
    n(t)\in B\ind{S}(\mathbb R^2)\ \forall\; t\in\mathbb R
        \label{equ:ntspace}
\end{equation}
and in particular: either Assumption \ref{assumptionS} holds and $\norm{\vu{T}n(t,x)}_\infty<1$, or $\norm{\vu{T}n(t,x)}_\infty<1/2$. The velocity function $v(p)$ is also bounded with respect to the measure induced by $n(t)$, as $\sup_{x,p}v(p)n_0(X_0(\hat X(t,x,p),p),p) \leq \sup_{x,p}v(p)n_0(x,p)$. Thus we have (as before, seeing $v(p)$ as a function of $t,x,p$ that is independent of $t,x$)
\begin{equation}\label{equ:vspacent}
	n\in B(\mathbb R^3,\abs{v}).
\end{equation}
In particular
\begin{equation}
    v\in B(\mathbb R,n(t,x))\ \forall\ t,x\in\mathbb R
\end{equation}
and
\begin{equation}\label{equ:vboundnt}
    \norm{v}_\infty^{n}\leq \norm{v}_\infty^{n_0}.
\end{equation}

The above means that we may dress with respect to $n(t,x)$ for any $t,x\in\mathbb R$ (see Definition \ref{def:dressing_general}), that $1\upd{dr}(t,x,p)$ is bounded as in Lemma \ref{lem:dressing_boundedness_of_1dr}, and that $v\upd{dr}(t,x,p)$ exists for all $t,x,p\in\mathbb R$.

\subsection{Continuity and differentiability}\label{ssectcontdiff}

We now establish the basic continuity and differentiability properties of $\hat X$ and $N$.  In order to discuss differentiability, we need to introduce some basic mathematical concepts. Recall that a function $f(x)$ (which may also depend on other variables) that is Lipschitz continuous is absolutely continuous, hence a.e.~differentiable. Further, one may extend its a.e.~derivative to a measurable function everywhere defined, which integrates to $f(x)$. This extension is a Radon-Nikodym derivative of the measure $\dd{f(x)}$ with respect to $\dd{x}$: the Radon-Nikodym derivative exists as absolute continuity of the function implies absolute continuity of the measure; it is uniquely defined on a dense subset; and we can extend it by setting it on the complement (a subset of measure zero) to a function that integrates to zero. Of course, this extension is not unique.

For our purposes, it will be convenient to consider an appropriate extension of the Radon-Nikodym derivative. We will denote it by $\frac{\check\dd{} f(x)}{\check\dd{} x}$, and refer to this simply as the {\em extended derivative}. More precisely, suppose that the function $\mathbb R\ni x\mapsto f(x)\in \mathbb R$ is absolutely continuous, and that its Radon-Nikodym derivative, defined on some dense subset $\Lambda\in\mathbb R$, is bounded $\sup_{x\in\Lambda}\abs{\frac{\dd{} f(x)}{\dd{} x}}<\infty$. Then we use the notation
\begin{equation}
	\frac{\check\dd{} f(x)}{\check\dd{} x}
\end{equation}
to mean {\em a choice of the extension to $\mathbb R$ of the Radon-Nikodym derivative of the measure $\dd{f(x)}$ with respect to $\dd{x}$, with
\begin{equation}\label{conditionintegral}
	f(x) - f(x') = 
	\int_{x'}^x \dd{y} \frac{\check\dd{} f(y)}{\check\dd{} y}\quad \forall\; x,x'\in \mathbb R,
\end{equation}
such that}
\begin{equation}\label{conditionbound}
	\sup_{x\in \mathbb R} \abs{\frac{\check\dd{} f(x)}{\check\dd{} x}}<\infty.
\end{equation}
This choice always exists. We extend the concept to partial derivatives, $\frac{\check\partial f(s,x)}{\check\partial x}$, in the natural fashion, requiring uniform boundedness on all variables.

We recall that the composition of Lipschitz functions is also Lipschitz, and that the a.e.~derivative of the composition of Lipschitz functions is obtained by the chain rule. Two basic results will be needed:
\begin{lemma}\label{lem:Lipbasic}
If $\mathbb R\ni x\mapsto f(x) \in\banachspace$ is Lipschitz, then
\begin{equation}
	\Big\{(x,p)\mapsto \frac{\check \partial f(x,p)}{\check\partial x}\Big\} \in B(\mathbb R^2).
\end{equation}
\end{lemma}
\begin{proof} By the Lipschitz property, say with Lipschitz constant $C>0$, wherever the derivative exists it satisfies
\begin{equation}
	\abs{\frac{ \partial f(x,p)}{\partial x} }\leq C.
\end{equation}
Then the extended derivative $\frac{\check \partial f(x,p)}{\check\partial x}$ is uniformly bounded, hence the result follows.
\end{proof}

\begin{lemma} \label{lem:LipT}
If $\mathbb R\ni x\mapsto f(x) \in\banachspace$ is Lipschitz, then $\mathbb R\ni x\mapsto \vu{T} f(x) \in\banachspace$ is Lipschitz and
\begin{equation}\label{equ:LipT}
	\frac{\check\partial\;\vu{T} f(x,p)}{\check\partial x} = \vu{T} \frac{\check \partial f(x,p)}{\check\partial x}\quad \forall\; x,p\in\mathbb R
\end{equation}
(recall that, in our convention, $\vu{T}$ acts on the momentum variable $p$ when it is explicit.)
\end{lemma}
\begin{proof} Clearly $\vu{T} f(x)$ is well defined for every $x\in\mathbb R$. The Lipschitz property of the function $x\mapsto f(x)$ means that there exists $C>0$ such that $\norm{f(x) - f(x')}_\infty \leq C\abs{x-x'}$ for all $x,x'\in\mathbb R$. Hence $\norm{\vu{T} f(x) - \vu{T} f(x')}_\infty \leq \norm{\vu{T}}\ind{op} C\abs{x-x'}$ for all $x,x'\in\mathbb R$ so that $x\mapsto \vu{T} f(x)$ is Lipschitz. Further, $\frac{\check \partial f(x,p)}{\check\partial x}$ is uniformly bounded on $\mathbb R^2$ by Lemma \ref{lem:Lipbasic}, hence so is $\vu{T}\frac{\check \partial f(x,p)}{\check\partial x}$.

Uniform boundedness of $\frac{\check \partial f(x,p)}{\check\partial x}$ allows us to use Fubini's theorem (first step) to get
\begin{equation}
	\int_{x'}^{x} \dd{y}\vu{T} \frac{\check \partial f(y,p)}{\check\partial y} =
	\vu{T} \int_{x'}^{x} \dd{y} \frac{\check \partial f(y,p)}{\check\partial y} = \vu{T}f(x,p) - \vu{T}f(x',p)
	= \int_{x'}^{x} \dd{y} \frac{\check \partial \; \vu{T} f(y,p)}{\check\partial y}
\end{equation}
for every $x,x',p\in\mathbb R$, where in the last step we used the fact that $x\mapsto \vu{T} f(x,p)$ is Lipschitz (hence $\vu{T} f(x)$ is absolutely continuous in $x$ for every $p$). As $\vu{T}\frac{\check \partial f(x,p)}{\check\partial x}$ is uniformly bounded, the above equation implies that we may choose the extended partial derivative of $\vu{T} f(x,p)$ as per \eqref{equ:LipT}.
\end{proof}

It turns out that we can evaluate the extended derivatives of $\hat X$ and $N$; and further that $\hat X(t,x,p)$ is strictly increasing in $x$, hence a homeomorphism of $\mathbb R$, and that $N(t,x,p)$ is non-decreasing in $x$. The following theorem, which establishes these, is at the heart of our main results: it will imply the existence of a solution to the GHD equation in potential form. It will also lead, by a simple recursive argument on the differential equations for $\hat X$ and $N$, to their differentiability properties for differentiable initial data. It is perhaps surprising that so little is required of the seed occupation function $n_0$ in order for this to be true.

Below, for $f\in B(\mathbb R,n(t,x))$, the expression $f\upd{dr}(t,x,p)$ is the dressing with respect to $n(t,x)$, that is $f\upd{dr}(t,x,p) = (1-\vu{T}n(t,x))^{-1}f(p)$.

\begin{theorem}[Continuity and a.e.~differentiability] \label{theo:XhatNbasic}
The functions $\hat X(t,x,p)$ and $N(t,x,p)$ vanish at the origin of space-time,
\begin{equation}\label{equ:NKvanish}
        \hat X(0,0,p)=0,\quad N(0,0,p)=0\quad \forall\;p\in\mathbb R.
\end{equation}
Further, $\hat X,N:\mathbb R^2\to \banachspace$ are Lipschitz continuous both with respect to the first ($t$) and second ($x$) arguments (uniformly in the other argument). In addition, for every $t,p\in\mathbb R$, $\hat X(t,\cdot,p)$ is ($t,p$-uniformly) bi-Lipschitz, hence strictly increasing and a homeomorphism of $\mathbb R$, and $N(t,\cdot,p)$ is non-decreasing. Finally, for every $t,x,p\in\mathbb R$,
    \begin{equation}\label{equ:xderivatives}
        \frac{\check\partial{N(t,x,p)}}{\check\partial{x}} = n(t,x,p)1\upd{dr}(t,x,p),\quad
        \frac{\check\partial{\hat X(t,x,p)}}{\check\partial{x}} = 1\upd{dr}(t,x,p)
    \end{equation}
and
    \begin{equation}\label{equ:tderivatives}
        \frac{\check\partial{N(t,x,p)}}{\check\partial{t}} = -n(t,x,p)v\upd{dr}(t,x,p),\quad
        \frac{\check\partial{\hat X(t,x,p)}}{\check\partial{t}} = -(v\upd{dr}(t,x,p) - v(p)).
    \end{equation}
\end{theorem}
\begin{proof}
Eqs.~\eqref{equ:NKvanish} hold as $\hat X(0,0,p)=0$ is a solution to the fixed-point problem \eqref{fp2proof} (recall that $\hat N_0(0,p)=0$ by its definition \ref{equ:N0_def}), and this solution is unique by Theorem \ref{theo:fixedpoint}; and $N(0,0,p)=0$ then follows form its definition \eqref{equ:defN}.

For Lipschitz continuity, we show that
for every $t,x,t',x',p\in\mathbb R$,
\begin{equation}\label{equ:KtxLipschitz}
    \abs{\hat X(t,x,p)-\hat X(t',x',p)}
    \leq
    \abs{x-x'}\frac{1}{1-\norm{\vu{T}n_0}\ind{op}}
    + \abs{t-t'}\norm{v}_\infty^{n_0}\frac{\norm{\vu{T}}\ind{op}}{1-\norm{\vu{T}n_0}\ind{op}}.
\end{equation}
and
\begin{equation}\label{equ:NtxLipschitz}
    \abs{N(t,x,p)-N(t',x',p)}
    \leq
    \abs{x-x'}\frac{\norm{n_0}_\infty}{1-\norm{\vu{T}n_0}\ind{op}}
    + \abs{t-t'}\norm{v}_\infty^{n_0}\Bigg(\frac{\norm{\vu{T}}\ind{op}\norm{n_0}_\infty}{1-\norm{\vu{T}n_0}\ind{op}} + 1\Bigg)
\end{equation}
For this purpose, we evaluate from \eqref{equ:GHcontraction} and \eqref{equ:GtxG00},
\begin{align}
    & \abs{\hat X(t,x,p)-\hat X(t',x',p)} \nonumber\\
    &\quad \leq \abs{G_{t,x}[\hat X(t,x)](p)-G_{t,x}[\hat X(t',x')](p)}
    +\abs{G_{t',x'}[\hat X(t',x')](p) - G_{t,x}[\hat X(t',x')](p)}\nonumber\\
    &\quad \leq 
    \norm{\vu{T}n_0}\ind{op} \norm{\hat X(t,x)-\hat X(t',x')}_\infty
    +
    \abs{x-x'}
        + \abs{t-t'}\,\norm{\vu{T}}\ind{op}\norm{v}_\infty^{n_0}
\end{align}
Therefore, as $\norm{\vu{T}n_0}\ind{op}<1$, taking the supremum on $p$, we get that $\norm{\hat X(t,x)-\hat X(t',x')}_\infty$ is bounded by the right-hand side of \eqref{equ:KtxLipschitz}, thus inequality \eqref{equ:KtxLipschitz} follows. Along with $\hat X(0,0,p)=0$, this implies $\hat X(t,x)\in B(\mathbb R)$ for all $t,x\in\mathbb R^2$, and that $\hat X:\mathbb R^2\to \banachspace$ is Lipschitz continuous both with respect to the first ($t$) and second ($x$) arguments, uniformly in the other argument.

For $N(t,x,p)$, we use its definition \eqref{equ:defN} along with \eqref{equ:N0_def},
\begin{align}
    \abs{N(t,x,p)-N(t',x',p)} &\leq
    \sup_{x\in\mathbb R}n_0(x,p) \Big(\abs{\hat X(t,x,p)-\hat X(t',x',p)} + \abs{v(p)}\abs{t-t'}\Big) \nonumber\\
    &\leq
    \norm{n_0}_\infty \abs{\hat X(t,x,p)-\hat X(t',x',p)} + \norm{v}^{n_0}_\infty \abs{t-t'}
\end{align}
which, in combination with \eqref{equ:KtxLipschitz},  leads to \eqref{equ:NtxLipschitz}. Along with $N(0,0,p)=0$, this implies $N(t,x)\in B(\mathbb R)$ for all $t,x\in\mathbb R^2$, and that $N:\mathbb R^2\to \banachspace$ is Lipschitz continuous both with respect to the first ($t$) and second ($x$) arguments, uniformly in the other argument.

Lipschitz continuity guarantees that $\hat X(t,x,p)$ and $N(t,x,p)$ are absolutely continuous in $t$ and $x$, fixing the other variables. Hence, they are a.e.~differentiable, and we may construct the associated extended derivatives, as discussed above.

We will now show the differential equations \eqref{equ:xderivatives}, \eqref{equ:tderivatives}. Once these are shown, then Lemma \ref{lem:dressing_boundedness_of_1dr} along with \eqref{equ:ntspace} imply
\begin{align}
    \label{equ:Ktincrease}
        \frac{\hat X(t,x,p) - \hat X(t,y,p)}{x-y} \geq R(\norm{\vu{T}n_0}\ind{op})\quad \forall\ t,x,y,p\in\mathbb R
\end{align}
and hence $\hat X(t,\cdot,p)$ is $t,p$-uniformly bi-Lipschitz, therefore it is strictly increasing and a homeomorphism of $\mathbb R$ for all $t,p\in\mathbb R$. Likewise, Lemma \ref{lem:dressing_boundedness_of_1dr} along with \eqref{equ:ntspace} imply     \begin{align}
\label{equ:Ntincrease}
        N(t,x,p) - N(t,y,p)\geq 0\quad \forall\ t,x,y,p \in \mathbb R\ :\ x\geq y
\end{align}
and $N(t,\cdot,p)$ is non-decreasing. This will complete the proof of the theorem.

The Lispchitz continuity statements are that $\mathbb R \ni x\mapsto \hat X(t,x), N(t,x)\in B(\mathbb R)$ are Lipschitz for every $t$, and $\mathbb R \ni t\mapsto \hat X(t,x), N(t,x)\in B(\mathbb R)$ are Lipschitz for every $x$. Hence from Lemma \ref{lem:Lipbasic} their extended derivatives are uniformly bounded over the differentiated variable and $p\in\mathbb R$. In particular,  we have
\begin{equation}\label{equ:Kspaceproof}
	\Big\{(x,p)\mapsto \frac{\check\partial \hat X(t,x,p)}{\check\partial x}\Big\}\in B(\mathbb R^2)\ \forall\,t,\quad 
	\Big\{(t,p)\mapsto \frac{\check\partial \hat X(t,x,p)}{\check\partial t}\Big\}
	\in B(\mathbb R^2)\ \forall\,x.
\end{equation}
Note that for these extended derivatives in $x$ and $t$, which exist but are not unique, we assume only, for now, that they satisfy \eqref{equ:Kspaceproof} (this assumption can always be fulfilled as follows from Lemma \ref{lem:Lipbasic}); we will determine them below.

From Lemma \ref{lem:N0_uniformly}, $\mathbb R\ni \hat x\mapsto \hat N_0(\hat x)\in B(\mathbb R)$ is Lipschitz. Hence with \eqref{equ:N0_def} we have
\begin{equation}\label{equ:N0derivaproof}
    \frac{\check\partial{\hat N_0(\hat x,p)}}{\check\partial{\hat x}}
    = n_0(X_0(\hat x,p),p),\quad \forall\, x,p;
\end{equation}
we can indeed choose the extended derivative as such, as the right-hand side is bounded on $\mathbb R^2$. Thus by the chain rule for Lipschitz functions and using \eqref{equ:defnt} and the definition \eqref{equ:defN}, we have (here shown simultaneously for both $x$ and $t$ derivatives)
\begin{equation}\label{intermediateNX}
    \frac{\check\partial{N(t,x,p)}}{\check\partial{x},\ \check\partial{t}}
    =\frac{\check\partial{\hat N_0(\hat X(t,x,p) - v(p)t,p)}}{\check\partial{x},\ \check\partial{t}}
    = n(t,x,p) \frac{\check\partial{ (\hat X(t,x,p) - v(p)t)}}{\check\partial{x},\ \check\partial{t}},\quad\forall\,t,x,p.
\end{equation} 
Again this is a good choice of extended derivatives for $N$, as both derivatives are uniformly bounded in $t,x,p$ by using \eqref{equ:ntbound1}, \eqref{equ:Kspaceproof} and \eqref{equ:vspacent} on the right-hand side. Let us use the function $K(t,x,p) = \hat X(t,x,p)-v(p)t$, which is notationally more convenient here.
By the Lipschitz continuity statement for $N$, Lemma \ref{lem:LipT} and \eqref{intermediateNX}, we have
\begin{equation}
    \frac{\check\partial \;\vu{T} N(t,x,p)}{\check\partial{x},\ \check\partial{t}}
    = \vu{T}\,n(t,x,p) \frac{\check\partial{ K(t,x,p)}}{\check\partial{x},\ \check\partial{t}}
    ,\quad\forall\,t,x,p.
\end{equation}
With the fixed-point problem \eqref{fp2proof} we then have
\begin{equation}\label{equ:Kpartialx}
    \frac{\partial K(t,x,p)}{\partial x} 
    =
    1 + \vu{T} \,n(t,x,p)\frac{\check\partial K(t,x,p)}{ \check\partial x}\quad
    \forall\, t,p,\ \mbox{a.e.}\ x
\end{equation}
and
\begin{equation}\label{equ:Kpartialt}
    \frac{\partial K(t,x,p)}{\partial t} 
    =
    -v(p)
    +\vu{T}\, n(t,x,p)\frac{\check\partial K(t,x,p)}{ \check\partial t}\quad
    \forall\, x,p,\ \mbox{a.e.}\ t.
\end{equation}
We have shown that there are extended derivatives such that \eqref{equ:Kspaceproof} hold, and this is what we used in order to derive the above partial derivative equations. Can we choose extended derivatives such that, in addition to \eqref{equ:Kspaceproof}, the extended versions of \eqref{equ:Kpartialx} and \eqref{equ:Kpartialt} hold for every $x,t$? We now show that this is the case. Indeed, by Definition \ref{def:dressing_general}, these are equivalent to $\frac{\check\partial K(t,x,p)}{\check\partial x}  = 1\upd{dr}(t,x,p)$ and $\frac{\check\partial K(t,x,p)}{\check\partial t}  = -v\upd{dr}(t,x,p)$, which do satisfy \eqref{equ:Kspaceproof} (recall that $v\upd{dr}(t,x,p)-v(p)$ is uniformly bounded). Hence the second equations of \eqref{equ:xderivatives} and \eqref{equ:tderivatives} follow. The first equations follow from \eqref{intermediateNX}.
\end{proof}


We emphasise that the extended derivatives \eqref{equ:xderivatives} and \eqref{equ:tderivatives} are, in fact, all in $B(\mathbb R^3)$. We also note that the homeomorphism property of $\hat X(t,\cdot,p)$ proved in Theorem \ref{theo:XhatNbasic} implies a strengthening of \eqref{equ:ntbound1}:
\begin{equation}\label{equ:ntbound2}
	\norm{n(t)}_\infty = \norm{n}_\infty = \norm{n_0}_\infty
\end{equation}
for all $t\in\mathbb R$.

With the above we can now obtain the differentiability class of $n$ from the differentiability class of $n_0$ (see Subsection \ref{ssect:notations}), by a simple recursive argument using the differential equations \eqref{equ:xderivatives} and \eqref{equ:tderivatives}.
\begin{theorem}[Strong differentiability]\label{theo:differentiability} Let $r\in\mathbb N$ and define $\nu(x,p)= \abs{v(p)}+1$. If $\nu n_0 \in C^r\ind{p}(\mathbb R,\nu)$ (that is, both $n_0$ and $v n_0$ belong to $C^r\ind{p}(\mathbb R,\nu)\subset C^r\ind{p}(\mathbb R)$), then $\nu n\in C^r\ind{p}(\mathbb R^2)$ (that is, both $n$ and $v n$ belong to $C^r\ind{p}(\mathbb R^2)$). In particular, if $r\geq 1$, then $n\in C^{r-1}\ind{u}(\mathbb R)\subset C^{r-1}(\mathbb R\to \banachspace)$.
\end{theorem}
\begin{proof}
By Lemma \ref{lem:k0_uniformly}, $X_0\in C^{r+1}\ind{p}(\mathbb R)\subset C^{r}\ind{p}(\mathbb R)$, thus $\nu \hat n_0\in C^r\ind{p}(\mathbb R,\nu)$ (see Eq.~\eqref{equ:defnt}). We proceed by induction. By Theorem \ref{theo:XhatNbasic}, $\hat X\in C^0(\mathbb R^2\to\banachspace) = C^0\ind{u}(\mathbb R^2)\subset C^0\ind{p}(\mathbb R^2)$. Assume that $\hat X\in C^a\ind{p}(\mathbb R^2)$ for some $0\leq a\leq r$; this holds for $a=0$. Then by Eq.~\eqref{equ:defnt}, and the composition rule \eqref{equ:Crcomposition}, $\nu n\in C^a\ind{p}(\mathbb R^2)$. If $a=r$, the induction concludes here and the Lemma follows. Otherwise, from Lemma \ref{lem:dressing_differentiability}, with $f(t,x,p) = 1$ and $f(t,x,p) = v(p)$, we find that $1\upd{dr}$ and $v\upd{dr}-v$ belong to $C^a\ind{p}(\mathbb R^2)$. By Definition \ref{def:dressing_general}, they simultaneously belong to $B(\mathbb R^2\times \mathbb R)$. Hence by Theorem \ref{theo:XhatNbasic}, $\hat X\in C^{a+1}\ind{p}(\mathbb R^2)$. This completes the induction.
\end{proof}

\subsection{Initial condition}\label{ssectinitcond}

Finally, the last point to prove is that the initial condition $n(0,x,p)$ is nothing else but the seed occupation function $n_0(x,p)$. This gives $n_0$ its full meaning beyond simply an auxiliary function determining the fixed-point problem.
\begin{theorem}[Initial condition] \label{theo:initconditionK0} We have $\hat X(0,x,p) = \hat X_0(x,p)$, thus $n(0,x,p) = n_0(x,p)$, for all $x,p\in\mathbb R$.
\end{theorem}
\begin{proof}
For any function $f(\cdot,p)$ that is absolutely continuous for every $p\in\mathbb R$, we define $\tilde n_f(x,p) = n_0(X_0(f(x,p),p),p)$ (see Eq.~\eqref{equ:defnt}). We have $\sup_{x\in\mathbb R}\tilde n_f(x,p)\leq \sup_{x\in\mathbb R}n_0(x,p)$ and therefore $\norm{\vu{T}(\sup_{x\in\mathbb R}n_f(x,\cdot)}\ind{op}\leq \norm{\vu{T}(\sup_{x\in\mathbb R}n_0(x,\cdot)}\ind{op}$, so that $\tilde n_f\in B\ind{S}(\mathbb R^2)$, thus $n_f$ satisfy the same assumption \ref{assumptionA} as $n_0$. We define $1\upd{dr}[f](x,p)$ as the dressing of 1 with respect to $\tilde n_f(x)$. Note that $1\upd{dr}[\hat X_0](x,p) = 1_0\upd{dr}(x,p)$ (see \eqref{equ:K0definition}), and $1\upd{dr}[\hat X(t)] = 1\upd{dr}(t,x,p)$. We show that the (highly non-linear!) initial value problem
\begin{equation}\label{equ:toshowinit}
	\frac{\check\partial{f(x,p)}}{\check\partial{x}} = 1\upd{dr}[f](x,p),\quad  f(0,p) = 0
\end{equation}
has a unique solution. By Theorem \ref{theo:XhatNbasic}, $\hat X(0,x,p)$ solves this, and by the definition \eqref{equ:K0definition}, $\hat X_0(x,p)$ also does, whence their equality.

Integrating \eqref{equ:toshowinit} and using Lemma \ref{lem:dressing_boundedness_of_1dr} we find that $\mathbb R \ni x\mapsto f(x)\in \banachspace$, and that this function is Lipschitz continuous. Hence from \eqref{equ:N0bounded} the same properties hold for $x\mapsto \tilde N_f(x)$ defined as $\tilde N_f(x,p) = \hat N_0(f(x,p),p)$. Thus by Lemmas \ref{lem:Lipbasic} and \ref{lem:LipT}, definition \eqref{equ:N0_def}, and the chain rule,
\begin{equation}
	\frac{\check\partial \;\vu{T}\tilde N_f(x,p)}{\check\partial{x}} =
	\vu{T}\, \tilde n_f(x,p)1\upd{dr}[f](x,p) ,\quad \forall\,x,p.
\end{equation}
Consider $L(x,p) := f(x,p) - (x+\vu{T}\tilde N_f(x,p))$; by the above this is Lipschitz continuous  $\mathbb R \ni x\mapsto  L(x)\in \banachspace$, and we have
\begin{equation}
	\frac{\check\partial{L(x,p)}}{\check\partial{x}} =
	1\upd{dr}[f](x,p)
	-
	(1 + \vu{T} \tilde{n}_f(x) 1\upd{dr}[f](x,p)).
\end{equation}
By Definition \ref{def:dressing_general}, the right-hand side vanishes. As $L(0,p)=0$, then $L(x,p)=0$. Therefore, $f(x,p) = x+\vu{T}\tilde N_f(x,p) = G_{0,x}[f(x,\cdot)](p)$ (see definition \eqref{equ:mapG}) for all $x,p\in\mathbb R$. By Theorem \ref{theo:fixedpoint}, $f$ is unique.
\end{proof}

\section{Existence and uniqueness of solutions to the GHD equations}\label{sec:existenceuniqueness}

We now show the existence and uniqueness of the solution to the GHD equation, both in its weak form when not assuming differentiability of the initial condition, and in its strong form when assuming it. The results of the previous section give us the framework to obtain the existence; what remains to be done is to connect them with the GHD equation more explicitly, and show uniqueness. We divide this section into two parts: we first discuss the weak forms of the GHD equation, where we introduce its ``potential form" and how it connects with its more standard weak statement; and we then impose differentiability and obtain the strong form.

We recall from Subsection \ref{ssect:reviewsetup} that we assume given a scattering shift $T(p,q)$ with \eqref{equ:condTop}, and a velocity function $v(p)$. We are looking for the existence and uniqueness of the GHD equation given $T$ and $v$. In this section, all additional assumptions on the other functions involved will be stated explicitly when needed.

We first note that a natural space in which to look for an occupation function solving the GHD equation is that of $n$ defined in \eqref{equ:defnt}: this is the space $B\ind{S}(\mathbb R^3)\cap B(\mathbb R^3,\abs{v})$ (see Eqs.~\eqref{equ:ntspace} and \eqref{equ:vspacent}), a subset of $B\ind{T}(\mathbb R^3)\cap B(\mathbb R^3,\abs{v})$. Recall that the latter is the space of non-negative real functions with $\norm{\vu{T}(\sup_{(t,x)\in\mathbb R^3}n(t,x,\cdot))}\ind{op}<1$ and $\sup_{(t,x,p)\in\mathbb R^3}\abs{v(p)}\,n(t,x,p)<\infty$, while the former has a possibly more stringent first inequality, see \eqref{equ:defBS}.

Second, we note that the relation between an occupation function $n(t,x,p)$ and its corresponding particle density $\rho\ind{p}(t,x,p)$ given in  \eqref{rhos1dr} is a bijection (onto its image) and gives a non-negative function $\rho\ind{p}(t,x,p)\geq 0\;\forall t,x,p$, and that the effective velocity defined by \eqref{veffvdr} is the unique solution to the integral equation \eqref{equ:intro_GHD_conservation_velocity}. Thus, with this, it is indeed sufficient to discuss the GHD equation in terms of the occupation function; statements in term of the particle density will follow. We formalise this in the following lemma.
\begin{lemma} The map $B\ind{S}(\mathbb R^3)\cap B(\mathbb R^3,\abs{v})\ni n\mapsto \rho\ind{p}$ given by
\begin{equation}\label{equ:rhopnproof}
	2\pi \rho\ind{p}(t,x,p) = n(t,x,p) 1\upd{dr}(t,x,p),\quad
\end{equation}
is injective (hence bijective onto its image). The result satisfies $\rho\ind{p}(t,x,p)\geq 0\;\forall t,x,p\in\mathbb{R}$. Further, the function
\begin{equation}\label{equ:veffnproof}
	v\upd{eff}(t,x,p) = \frac{v\upd{dr}(t,x,p)}{1\upd{dr}(t,x,p)}
\end{equation}
is the unique element of $B(\mathbb R^3,n1\upd{dr})$ which solves \eqref{equ:intro_GHD_conservation_velocity}.
\end{lemma}
\begin{proof} As $\rho\ind{p}(t,x,\cdot)$ is bounded, then $\int_{\mathbb R}\dd{q} T(p,q)\rho\ind{p}(t,x,q)$ exists and is finite, and, using \eqref{equ:rhopnproof} and Definition \ref{def:dressing_general}, the function
\begin{equation}
	\rho\ind{s}(t,x,p) = \frac1{2\pi}
	+\int_{\mathbb R}\dd{q} T(p,q)\rho\ind{p}(t,x,q)
\end{equation}
satisfies $2\pi\rho\ind{s}(t,x,p) = 1\upd{dr}(t,x,p)$. By Lemma \ref{lem:dressing_boundedness_of_1dr}, we have $2\pi\rho\ind{s}(t,x,p)>0$ for all $t,x,p$ (uniformly), and therefore $\rho\ind{p}(t,x,p)\geq 0$. Thus, with \eqref{equ:rhopnproof} we find $n = \rho\ind{p}/\rho\ind{s}$, which is the required inverse.

Finally, in one direction, as $n\in B(\mathbb R^3,\abs{v})$ we have $v\in B(\mathbb R^3,n)$ and hence by Definition \ref{def:dressing_general}, $v\upd{dr} \in B(\mathbb R^3,n)$ thus from the definition \eqref{equ:veffnproof}, $v\upd{eff} \in B(\mathbb R^3,n1\upd{dr})$. Putting \eqref{equ:veffnproof} in \eqref{equ:intro_GHD_conservation_velocity} we use
\begin{align}
v(p) + \int \dd{q}\varphi(p,q)\rho\ind{p}(t,x,q)v\upd{eff}(t,x,q) &= v(p) + \int \dd{q}T(p,q)n(t,x,q)v\upd{dr}(t,x,q)\nonumber\\
&= v\upd{dr}(t,x,p)\label{proofwww1}
\end{align}
and
\begin{align}
&\Big(1+ \int \dd{q}\varphi(p,q)\rho\ind{p}(t,x,q)\Big)v\upd{eff}(t,x,p)\nonumber\\
&\hspace{2cm}= \Big(1 + \int \dd{q}T(p,q)n(t,x,q)1\upd{dr}(t,x,q)\Big)v\upd{eff}(t,x,p)\nonumber\\
&\hspace{2cm}= 1\upd{dr}(t,x,p) v\upd{eff}(t,x,p) \nonumber\\
&\hspace{2cm}= v\upd{dr}(t,x,p)\label{proofwww2}
\end{align}
and get equality, showing that \eqref{equ:veffnproof} is a solution to \eqref{equ:intro_GHD_conservation_velocity}. In the opposite direction, suppose $v\upd{eff} \in B(\mathbb R^3,n1\upd{dr})$ solves \eqref{equ:intro_GHD_conservation_velocity}. We again consider the two quantities on the left-hand sides of \eqref{proofwww1} and \eqref{proofwww2}, which must be equated. Writing $v\upd{eff} = w/1\upd{dr}$ we have $w\in B(\mathbb R^3,n)$, and by the calculations above, we get
\begin{equation}
	v(p) + \int \dd{q}T(p,q)n(t,x,q)w(t,x,q)
	= w(t,x,p).
\end{equation}
By Definition \ref{def:dressing_general}, we conclude that $w(t,x,p) = v\upd{dr}(t,x,p)$.
\end{proof}

\subsection{Weak formulations of the GHD equation}\label{ssectweak}

We now express broadly three different ways of formulating the GHD equation ``weakly" -- without having to assume differentiability of the particle density and its current. The goal is two-fold: first to characterise solutions for initial conditions that only satisfy appropriate weak properties of boundedness, and second to obtain formulations which are potentially more amenable to proofs from microscopic models (we discuss this aspect in the conclusion). For such weak initial conditions, one can still impose constraints of various strengths on the dynamics itself: one may look at families of functions on space-time, with the same given initial condition, that are more or less large, all ``morally" satisfying the GHD equation. We will require the family to be small enough in order to have uniqueness; thus we will need strong enough constraints on the dynamics. We will then consider how these different ways of expressing weak solutions are related to each other.

Here, it is sufficient to take $n\in B\ind{T}(\mathbb R^3)\cap B(\mathbb R^3,\abs{v})\supset B\ind{S}(\mathbb R^3)\cap B(\mathbb R^3,\abs{v})$, which by Definition \ref{def:dressing_general} guarantees that $1\upd{dr}(t,x,p)$ and $v\upd{dr}(t,x,p)$ exist and that $1\upd{dr}(t,x,p)$ and $n(t,x,p)v\upd{dr}(t,x,p)$ are uniformly bounded. But in fact, all definitions and the lemma below are very general: the momentum variable $p$ plays no direct role -- it can be seen as a fixed parameter -- and, fixing this parameter, we only need to use uniform boundedness in space-time of the density $q(t,x) = n(t,x,p)1\upd{dr}(t,x,p)$ and current $j(t,x) = n(t,x,p)v\upd{dr}(t,x,p)$. That is, one could reformulate the definitions as weak formulations of a generic continuity equation $\partial_t q(t,x) + \partial_x j(t,x)=0$ for uniformly bounded $q,j\in B(\mathbb R^2)$, and the lemma as the equivalence of these formulations. We keep the special notation adapted to GHD as it makes the specialisation to this case clearer.

Perhaps the most convenient weak formulation of the GHD equation is its potential form, which we now define, both in its ``complete" and ``partial" variants; the complete one is that with the strongest constraint, for which existence and uniqueness will follow.

\begin{definition}[Potential form of the GHD equation]\label{def:potential} We say that a function $n\in B\ind{T}(\mathbb R^3)\cap B(\mathbb R^3,\abs{v})$ satisfies the GHD equation in complete (resp.~partial) potential form, if there exists a potential, that is a function $\mathbb R^3 \ni(t,x,p)\mapsto N(t,x,p)\in\mathbb R$ with the properties that for every $p\in\mathbb R$, $N(t,x,p)$ is absolutely continuous in $x\in\mathbb R$ uniformly for $t\in\mathbb R$, and is absolutely continuous in $t\in\mathbb R$ uniformly for $x\in\mathbb R$; with $N(0,0,p)=0$; such that for every $p\in\mathbb R$,
    \begin{equation}\label{equ:Ndifferentialequations}
        \frac{\partial{N(t,x,p)}}{\partial{x}} = n(t,x,p)1\upd{dr}(t,x,p),\quad
        \frac{\partial{N(t,x,p)}}{\partial{t}} = -n(t,x,p)v\upd{dr}(t,x,p).
    \end{equation}
The first equation is for every (resp.~a.e.) $t\in\mathbb R$ and a.e.~$x\in\mathbb R$, and the second, for every (resp.~a.e.) $x\in\mathbb R$ and a.e.~$t\in\mathbb R$.
\end{definition}
Note that in both cases the potential is unique, which can be seen by integration. The only difference between the complete and partial variant is if the $x$-derivative gives the stated result for every $t$, or a.e.~$t$, and vice versa for the $t$ derivative. Clearly, if the potential is twice continuously differentiable in $t,x$, then the equality of mixed derivatives gives the GHD equation Eq.~\eqref{equ:intro_GHD_conservation}. 

Another very natural weak formulation of the GHD equation is that which represents the continuity relation in terms of integrals instead of derivatives. Formally, this is obtained by integrating the GHD equation over a region of space-time; the result, obtained simply by performing one integral against one derivative in each term, is that the integration over the boundary of the region of the perpendicular current is zero (Stokes theorem). We do not assume differentiability, but the resulting vanishing of the contour integral can still be written, and is a common way of expressing the continuity equation. Here, it is sufficient to concentrate on rectangular contours (from which large classes of contours can be obtained). We will refer to this as the ``integral form".
\begin{definition}[Integral form of the GHD equation] \label{def:integral}
We say that a function $n\in B\ind{T}(\mathbb R^3)\cap B(\mathbb R^3,\abs{v})$ satisfies the GHD equation in complete (resp.~partial) integral form, if
\begin{equation}\label{equ:integralform}\begin{aligned}
	&\int_{x_1}^{x_2} \dd{x}
	\Bigg(
	n(t_2,x,p)1\upd{dr}(t_2,x,p)-
	n(t_1,x,p)1\upd{dr}(t_1,x,p)
	\Bigg)\\ &\qquad
	+
	\int_{t_1}^{t_2}\dd{t}
	\Bigg(
	n(t,x_2,p)v\upd{dr}(t,x_2,p)-
	n(t,x_1,p)v\upd{dr}(t,x_1,p)
	\Bigg) = 0
	\end{aligned}
\end{equation}
for every $p\in\mathbb R$ and every (resp.~a.e.) $(x_1,x_2,t_1,t_2)\in\mathbb R^4$.
\end{definition}

Finally, a perhaps more standard weak formulation is that which involves integrations against rapidly decreasing (Schwartz) functions (see e.g.~\cite{rudin}). This we will refer to as the ``weak form" of the GHD equation; it does not have a ``complete" and ``partial" variant; or more precisely, effectively it only have the latter one.
\begin{definition}[Weak form of the GHD equation] \label{def:weak}
We say that a function $n\in B\ind{T}(\mathbb R^3)\cap B(\mathbb R^3,\abs{v})$ satisfies the GHD equation in weak form, if for every rapidly decreasing function $f:\mathbb R^2\to\mathbb R$ and every $p\in\mathbb R$, we have
\begin{equation}\label{equ:weakform}
\int_{\mathbb R^2} \dd{t}\dd{x}\Big(
\partial_t f(t,x)\, n(t,x,p)1\upd{dr}(t,x,p)
+
\partial_x f(t,x)\, n(t,x,p)v\upd{dr}(t,x,p)\Big)
=0.
\end{equation}
\end{definition}

Clearly any complete variant imply its associated partial variant. Also, because the density and current are uniformly bounded, it is relatively elementary to show that the above three partial variants are in fact equivalent. We express this in the following lemma (this can be seen as a ``weak version" of the Poincar\'e lemma, which stipulates the existence and uniqueness of a potential for a conservation law). The proof is provided in Appendix \ref{appweak}.

\begin{lemma}[Equivalence of weak statements]\label{lem:equivalenceweak}
Let $n\in B\ind{T}(\mathbb R^3)\cap B(\mathbb R^3,\abs{v})$. Then concerning the GHD equation for $n$: complete potential form $\Leftrightarrow$ complete integral form; and partial potential form $\Leftrightarrow$ partial integral form $\Leftrightarrow$ weak form.
\end{lemma}

\begin{remark}
In fact, another natural weak formulation involves integrating against a test function in space only: for every rapidly decreasing function $f:\mathbb R\to\mathbb R$ and every $p\in\mathbb R$, one asks for the function $\mathbb R\ni t\mapsto \int \dd{x} f(x)n(t,x,p)1\upd{dr}(t,x,p)$ to be absolutely continuous, with
\begin{equation}\label{equ:weakform2}
	\frac{\partial }{\partial t} \int \dd{x} f(x)n(t,x,p)1\upd{dr}(t,x,p) = \int \dd{x} \partial_x f(x)\,n(t,x,p)v\upd{dr}(t,x,p)
\end{equation}
for a.e.~$t\in\mathbb R$. The above three weak formulations in their partial variants are also equivalent to this one, and there is similarly a complete variant (we omit the proof).
\end{remark}

\subsection{Existence and uniqueness theorems}

With the fixed-point problem of Subsection \ref{ssectfixedpoint}, and the weak formulations of the GHD equation in Subsection \ref{ssectweak}, we now have all ingredients in order to prove our main theorems: the existence and uniqueness of the solution to the GHD equation in its weak formulations (Theorem \ref{theo:potentialform}), and the existence and uniqueness in strong, differentiable form, with controlled differentiability properties for all times (Theorem \ref{theo:differentiableform}).

For the weak formulations, thanks to Lemma \ref{lem:equivalenceweak}, for showing existence it is sufficient to choose one of the complete variants, as existence for it will imply existence for all other weak formulations (complete and partial variants). We choose the potential form. For uniqueness, we only have results for the complete variants. Clearly the solutions in partial variants cannot be strictly unique as constraints on the dynamics are only ``almost everywhere". One may hope to have ``almost unique" solutions, however we leave this for future investigations. In any case, Theorem \ref{theo:potentialform} establishes that the complete variants are constraining enough on the dynamics to give, for any weak initial condition with appropriate bounds, existence and uniqueness of the GHD equation.

When controlling differentiability, we use the spaces introduced in \eqref{equ:defCp} and \eqref{equ:defCu}: these ask for pointwise, resp.~uniform, differentiability, along with uniform bounds on all derivatives of order 1 and more.

Here and below the dressing of $f\in B(\mathbb R,n(t,x))$ is with respect to $n(t,x)$ (Definition \ref{def:dressing_general}). 

\begin{theorem}[Existence and uniqueness, potential form]\label{theo:potentialform}
Assume \eqref{assumptionA}, \eqref{assumptionV}. There exists an essentially unique $n\in B\ind{S}(\mathbb R^3)\cap B(\mathbb R^3,\abs{v})$ such that $n(0,x,p) = n_0(x,p)$ which satisfies the GHD equation in complete potential form (Definition \ref{def:potential}). Further, let $\hat N_0,\,X_0$ be defined as in Subsection \ref{ssectseed}. Then the function $\hat X(t,x,p) = x + \vu{T} N(t,x,p)$, see Eq.~\eqref{XhatN}, is the unique solution to the fixed-point problem \eqref{fp2proof}, and the function $n$ is given by \eqref{equ:defnt} for a.e.~$t\in\mathbb R$ and all $x,p\in\mathbb R$.
\end{theorem}
\begin{proof}
Note that as both \eqref{assumptionA} and \eqref{assumptionV} are assumed to hold, we may indeed construct $\hat N_0$, $\hat X_0$ and $X_0$ as in Subsection \ref{ssectseed}, which we will use below.

{\bf Existence.}

Theorems \ref{theo:XhatNbasic} and \ref{theo:initconditionK0}, and Eqs.~\eqref{equ:ntspace}, \eqref{equ:vspacent}, show that $N$ and $n$ exist with the properties stated.

{\bf Uniqueness.}

Assume that there exists a potential $N$ as per Definition \ref{def:potential}.

{\em Lipschitz continuity of $N$.} By \eqref{equ:bounddressingn} we have the inequalities
\begin{equation}
	\abs{n(t,x,p)1\upd{dr}(t,x,p)}\leq\frac{\norm{n(t,x)}_\infty}{1-\norm{\vu{T}n(t,x)}\ind{op}}\leq\frac{\norm{n}_\infty}{1-\norm{\vu{T}n}\ind{op}},
\end{equation}
and
\begin{equation}
	\abs{n(t,x,p)v\upd{dr}(t,x,p)}\leq\frac{\norm{v}_\infty^{n(t,x)}}{1-\norm{\vu{T}n(t,x)}\ind{op}}\leq\frac{\norm{v}_\infty^{n}}{1-\norm{\vu{T}n}\ind{op}}.
\end{equation}
Thus integrating the first and second equation of \eqref{equ:Ndifferentialequations} over $x$ and $t$, respectively, we have that $(x,t)\mapsto N(t,x)\in\banachspace$ is Lipschitz continuous in $t$ uniformly for $x\in\mathbb R$, and in $x$ uniformly for $t\in\mathbb R$ (for functions of two variables $(t,x)$, we will refer to this simply as ``Lipschitz continuous").

{\em Definition of $\hat X$ and its differential equations.} Let us define $\hat X(t,x,p)$ by \eqref{XhatN}; note that $\hat X(0,0,p)=0$. Then by Lemma \ref{lem:LipT}, $(x,t)\mapsto \hat X(t,x)\in B(\mathbb R)$ is Lipschitz continuous, and for all $t,x,p$,
\begin{equation}
	\frac{\check\partial \hat X(t,x,p)}{\check\partial x}
	=
	1 + \vu T n(t,x)1\upd{dr}(t,x,p),\quad
	\frac{\check\partial \hat X(t,x,p)}{\check\partial t}
	=
	- \vu T n(t,x)v\upd{dr}(t,x,p).
\end{equation}
From Definition \ref{def:dressing_general}, we conclude that
\begin{equation}\label{equ:Kdifferentialequations}
        \frac{\check\partial{\hat X(t,x,p)}}{\check\partial{x}} = 1\upd{dr}(t,x,p),\quad
        \frac{\check\partial{\hat X(t,x,p)}}{\check\partial{t}} = -(v\upd{dr}(t,x,p) - v(p)).
\end{equation}
Therefore by \eqref{equ:K0definition} and the initial value $\hat X(0,0,p)=0$ as well as $n(0,x,p) = n_0(x,p)$, we find that $\hat X(0,x,p) = \hat X_0(x,p)$ for all $x,p$.

{\em Lipschitz continuity of $X$.} By integrating \eqref{equ:Kdifferentialequations},
\begin{equation}
	\hat X(t,x,p)-\hat X(s,y,p) = \int_y^x \dd{x'}1\upd{dr}(s,x',p)
	-\int_s^t \dd{t'} (v\upd{dr}(t',x,p) - v(p))
\end{equation}
and therefore by Lemma \ref{lem:dressing_boundedness_of_1dr} and the bound \eqref{equ:bounddressingnff}, for every $t,s,x,y,p$,
\begin{equation}\label{equ:ineqKproof}
	\abs{\hat X(t,x,p)-\hat X(s,y,p)}
	+
	\abs{t-s}
	\frac{\norm{\vu{T}}\ind{op}\norm{v}_\infty^{n}}{1-\norm{\vu{T}n}\ind{op}} \geq \abs{x-y} R(\norm{\vu Tn}\ind{op}).
\end{equation}
Hence taking $t=s$, we find that $\mathbb R \ni x\mapsto \hat X(t,x)\in\banachspace$ is bi-Lipschitz for every $t$ (uniformly). Thus $x\mapsto  \hat X(t,x,p)$ is strictly increasing and a homeomorphism of $\mathbb R$. Define $\hat x\mapsto X(t,\hat x,p)$ as its inverse, for every $t,\hat x,p$. Then $\mathbb R\ni \hat x\mapsto X(t, \hat x)\in\banachspace$ is bi-Lipschtiz uniformly in $t$. Also, using $\hat X(t,X(t,\hat x,p),p) = \hat x$ and \eqref{equ:ineqKproof} for $x = X(t,\hat x,p),\,y = X(s,\hat x,p)$, we have
\begin{equation}
	\abs{X(t,\hat x,p) - X(s,\hat x,p)}\,R(\norm{\vu Tn}\ind{op})\leq \abs{t-s}
	\frac{\norm{\vu{T}}\ind{op}\norm{v}_\infty^{n}}{1-\norm{\vu{T}n}\ind{op}}
\end{equation}
and as $R(\norm{\vu Tn}\ind{op})>0$ we find that $\mathbb R\ni t\mapsto X(t, \hat x)\in\banachspace$ is Lipschitz uniformly in $\hat x$.

{\em The function $\tilde N_0$.} Using the chain rule (for Lipschitz functions) to evaluate the left-hand side of $\frac{\check\partial{ \hat X(t,X(t,\hat x+v(p)t,p))}}{\check\partial{t}} = v(p)$ with \eqref{equ:Kdifferentialequations}, we obtain
\begin{equation}
	v\upd{dr}(t,x,p) = \frac{\check\partial{X(t,\hat x+v(p)t,p)}}{\check\partial{t}}
	1\upd{dr}(t,x,p)
\end{equation}
where $x = X(t,\hat x+v(p)t,p)$. Let us define $\tilde N(t,\hat x,p) = N(t,X(t,\hat x+v(p)t,p),p)$. By the chain rule again along with \eqref{equ:Ndifferentialequations},
\begin{align}
	\frac{\check\partial{\tilde N(t,\hat x,p)}}{\check\partial{t}}
	&=
	-n(t,x,p)v\upd{dr}(t,x,p)
	+
	\frac{\check\partial{X(t,\hat x+v(p)t,p)}}{\check\partial{t}}
	n(t,x,p)1\upd{dr}(t,x,p)
	\nonumber\\
	&=0.
\end{align}
Hence $\tilde N(t,\hat x,p) = \tilde N(0,\hat x,p)$ for every $t,\hat x,p$, and we define $\tilde N_0(\hat x,p) = \tilde N(0,\hat x,p)$.

{\em The fixed-point problem and uniqueness.} We note that, using the definition of $\tilde N$ as well as $n(0,x,p) = n_0(x,p)$ (assumption of the theorem),
\begin{equation}
    \tilde N_0(\hat x,p) = N(0,X(0,\hat x,p),p)
    = \int_0^{X(0,\hat x,p)} \dd{y} n_0(y,p)1\upd{dr}_0(y,p)
\end{equation}
and changing variable (by the homeomorphism established above) to $y = X(0,\hat y,p) = X_0(\hat y,p)$  (thanks to the equality $\hat X(0,x,p) = \hat X_0(x,p)$ established above), we obtain, from the change of measure \eqref{equ:measureXhat0},
\begin{equation}
    \tilde N_0(\hat x,p) 
    = \int_0^{\hat x} \dd{\hat y} n_0(X_0(\hat y,p),p)
    = \hat N_0(\hat x,p)
\end{equation}
where we used the definition \eqref{equ:N0_def}. Hence $N(t,x,p) = \hat N_0(\hat X(t,x,p)-v(p)t,p)$, and by the definition \eqref{XhatN} we obtain the fixed-point equation
\begin{equation}
    \hat X(t,x,p) = x + \vu{T}\hat N_0(\hat X(t,x,p)-v(p)t,p).
\end{equation}
By Theorem \ref{theo:fixedpoint}, $\hat X$ is unique. Further, by Theorem \ref{theo:XhatNbasic}
\begin{equation}
	\frac{\check\partial N(t,x,p)}{\check\partial x} \Big/
	\frac{\check\partial \hat X(t,x,p)}{\check\partial x}
	= n_0(X_0(\hat X(t,x,p)-v(p)t,p),p)
\end{equation}
hence from the first equations in \eqref{equ:Ndifferentialequations} and \eqref{equ:Kdifferentialequations}, we have $n(t,x,p) = n_0(X_0(\hat X(t,x,p)-v(p)t,p),p)$ for a.e.~$t\in\mathbb R$ and all $x,p\in\mathbb R$. This shows essential uniqueness of $n$, and the last sentence of the theorem.
\end{proof}

\begin{theorem}[Existence and uniqueness, differentiable form] \label{theo:differentiableform}
Assume \eqref{assumptionA}, \eqref{assumptionV}. Let $r\in\mathbb N=\{0,1,2,\ldots\}$ and define $\mathbb R^2\ni (x,p)\mapsto \nu(x,p) = \abs{v(p)}+1$ (which is thus independent of $x$). If $n_0,v n_0\in C^r\ind{p}(\mathbb R,\nu)$ (see Eq.~\eqref{equ:defCp}) then the unique $n\in B\ind{S}(\mathbb R^3)\cap B(\mathbb R^3,\abs{v})$ that solves the GHD equation in potential form satisfies $n,vn\in C\ind{p}^r(\mathbb R^2)$. In particular, if $r\geq 1$, there exists a unique $n\in B\ind{S}(\mathbb R^3)\cap B(\mathbb R^3,\abs{v})\cap C\ind{p}^1(\mathbb R^2)$ such that $n(0,x,p) = n_0(x,p)$, which solves the differential equation
\begin{equation}\label{equ:ghdproof}
        \frac{\partial}{\partial t}\big( n(t,x,p)1\upd{dr}(t,x,p)\big)
        +
        \frac{\partial}{\partial x} \big(n(t,x,p)v\upd{dr}(t,x,p)\big)
        =0,
\end{equation}
and the solution is $p$-uniformly continuously $r-1$ times differentiable, $n \in C^{r-1}(\mathbb R\to \banachspace)$.
\end{theorem}
\begin{proof} The first statement follows from Theorems \ref{theo:potentialform} and \ref{theo:differentiability}. For the second statement, by Theorem \ref{theo:differentiability} and Lemma \ref{lem:dressing_differentiability}, the quantities differentiated in \eqref{equ:ghdproof} are indeed continuously differentiable, and by the Poincar\'e lemma, \eqref{equ:ghdproof} is equivalent to the potential form. Thus existence and uniqueness follow from the first statement, and uniform continuous $(r-1)$-differentiability follows from Theorem \ref{theo:differentiability}.
\end{proof}

Recall how requiring $r$ times pointwise differentiability with uniformly bounded derivatives lead to $r-1$ times uniform differentiability. In particular, the case $r=\infty$ gives {\em smoothness of the result at all times}, from smoothness of the initial condition (along with uniformly bounded derivatives and uniform bounds involving the velocity function, as stated in the theorem). 

\begin{remark}
    Attempting to extend the proof of uniqueness to the partial potential form, one encounters two problems. The first is that the initial condition at $t=0$ might not be in the set where the first equation of \eqref{equ:Ndifferentialequations} holds. Naturally, we must impose that $t=0$ is in that set, otherwise the concept of initial condition is not meaningful. The second problem is deeper. As the dense set of $(x,t)$ where \eqref{equ:Ndifferentialequations} hold may depend on $p$, problems arise when applying the operator $\vu{T}$. The fact that this set may depend on $p$ can be traced back to the passage from the weak form to the partial integral form, where we do not necessarily have almost everywhere differentiability for a Lipschitz function taking values in $\banachspace$, because the latter space is not reflexive.
\end{remark}

\section{Consequences of the construction}\label{sec:consequences}

\subsection{Measure change and the characteristic function}

We now have a clear measure interpretation of the change of coordinate from $x$ to $\hat x = \hat X(t,x,p)$, thanks to absolute continuity and the specific results of Theorem \ref{theo:XhatNbasic},
\begin{equation}
	\dd \hat X(t,x,p) = 1\upd{dr}(t,x,p) \dd x\qquad \mbox{($t,p$ fixed)};
\end{equation}
and at any point $(t,x)$ where $n(t,x,p)=0$, $\hat X(t,x,p)$ is independent of $t$ (has zero extended $t$-derivative).

We denote the inverse of $\hat X(t,\cdot,p)$, which exists and is unique by Theorem \ref{theo:XhatNbasic}, as $X(t,\cdot,p)$:
\begin{equation}
	X(t,\hat x,p)\quad : \quad
	\hat X(t,X(t,\hat x,p),p) = \hat x.
\end{equation}
As $\hat X(t,\cdot,p)$ is $t,p$-uniformly bi-Lipschitz, so is its inverse, wherefore $X(t,\cdot,p)$ is absolutely continuous. Clearly, then, by the differentiation of the composition of Lipschitz functions, we have
\begin{equation}
	\frac{\check\partial{X(t,\hat x,p)}}{\check\partial{\hat x}} = \frac1{1\upd{dr}(t,X(t,\hat x,p),p)}
\end{equation}
which is finite and strictly positive (by Lemma \ref{lem:dressing_boundedness_of_1dr}). We may also define the characteristic function~\cite{BenGHD}
\begin{equation}\label{equ:u}
    u(t,x,p) = X(0,\hat X(t,x,p)-v(p)t,p),
\end{equation}
which is interpreted as the initial point of the trajectory of a quasi-particle with momentum $p$ ending at $x$ at time $t$, i.e.\ it satisfies \eqref{equ:u_eq}.
\begin{corol}[of Theorem \ref{theo:XhatNbasic}] The function $u(t,x,p)$ defined in \eqref{equ:u} is absolutely continuous in $x$ and $t$, hence a.e.~differentiable, and is a characteristic function for the GHD equation:
\begin{equation}\label{equ:u_eq}
	\frac{\check\partial{u(t,x,p)}}{\check\partial{t}}
	+ v\upd{eff}(t,x,p) \frac{\check\partial{u(t,x,p)}}{\check\partial{x}}
	= 0,\qquad u(0,x,p) = x
\end{equation}
where
\begin{equation}
	v\upd{eff}(t,x,p) = \frac{v\upd{dr}(t,x,p)}{1\upd{dr}(t,x,p)}.
\end{equation}
\end{corol}
In particular, this means that the characteristic curves for any given mode $p$ are non-crossing -- recall that, as per the differential equation and initial condition above, $u(t,x,p)$ represents the position at time 0 such that the characteristic curve is passes by the position $x$ at time $t$; as $u(t,x,p)$ is uniquely defined, the curves are non-crossing.

\subsection{Conserved quantities and entropy}
The GHD equation implies a huge number of conserved quantities. We will now show that for suitable functions $f$ or $g$ the following quantities are conserved during time:
\begin{align}
    \vu{S}_p[f](t) &= \tfrac{1}{2\pi}\int_{-\infty}^\infty\dd{x}f(n(t,x,p))1\upd{dr}(t,x,p)\\
    \vu{S}[g](t) &= \tfrac{1}{2\pi}\int_{-\infty}^\infty\dd{x}\dd{p}g(n(t,x,p),p)1\upd{dr}(t,x,p).
\end{align}
\begin{theorem}[Conserved entropies]\label{theo:conserved_entropies}
Assume the assumptions of Theorem \ref{theo:potentialform} and consider either $\vu{S}_p[f](0)$ exists for a function $f: \mathbb{R} \to \mathbb{R}$ and fixed $p$ or $\vu{S}[g](0)$ exists for a function $g: \mathbb{R}^2 \to \mathbb{R}$. Then $\vu{S}_p[f](t) = \vu{S}_p[f](0)$ or $\vu{S}[g](t) = \vu{S}[g](0)$.
\end{theorem}
\begin{proof}
This follows directly from 
\begin{align}
    \vu{S}_p[f](t) &= \tfrac{1}{2\pi}\int_{-\infty}^\infty\dd{x}f(n_0(X_0(\hat{X}(t,x,p)-v(p)t,p),p)\frac{\check\partial{\hat X(t,x,p)}}{\check\partial{x}}\\
    &= \tfrac{1}{2\pi}\int_{-\infty}^\infty\dd{\hat{x}}f(n_0(X_0(\hat{x},p),p),
\end{align}
which is independent of time and thus $\vu{S}_p[f](t) = \vu{S}_p[f](0)$. The proof for $\vu{S}[g](t)$ is similar.
\end{proof}
These quantities can be viewed as generalized entropies. Indeed the cases $g(n,p) = -n\log n$, $g(n,p) = \log n$, $g(n,p) = -n\log n - (1-n)\log(1-n)$ and $g(n,p) = -n\log n + (1+n)\log(1+n)$ correspond to the entropy of a system of classical particles, radiation, quantum fermions and quantum bosons respectively. In fact, hydrodynamic equations like the GHD equation are known to preserve entropy, however as GHD is agnostic to the particle statistic it has to conserve all of them. For instance it is possible to construct a classical integrable model with the same scattering shift as, for instance, the quantum Lieb-Liniger model~\cite{PhysRevLett.132.251602}.  

Entropies in general are only emergent conserved quantities, i.e. they are not conserved anymore beyond the hydrodynamic approximation. However, in the special case $f(n) = n$ they correspond to the actual microscopic conserved charges of the integrable model, i.e. for:
\begin{align}
    \vu{Q}_p(t) &= \int_{-\infty}^\infty\dd{x}\rho\ind{p}(t,x,p)\\
    \vu{Q}[h](t) &= \int_{-\infty}^\infty\dd{x}\dd{p}\rho\ind{p}(t,x,p)h(p)
\end{align}
we have:
\begin{theorem}[Conserved charges]\label{theo:conserved_charges}
Assume the assumptions of Theorem \ref{theo:potentialform} we have $\vu{Q}_p(t) = \vu{Q}_p(0)$. Alternatively, for any function $h:\mathbb{R} \to \mathbb{R}$ such that $\vu{Q}[h](0)$ exists we have $\vu{Q}[h](t) = \vu{Q}[h](0)$
\end{theorem}
\begin{proof}
Both follow directly from Theorem \ref{theo:conserved_entropies} by setting either $f(n) = n$ or $g(n,p) = nh(p)$.
\end{proof}

\section{Conclusion}\label{sec:conclu}
In this work we establish the first general proof of existence, uniqueness and regularity of global solutions to the GHD equation, for models with integrable scattering shift. It applies to a large class of initial states whose occupation function is smaller than a certain bound. In some quantum mechanical models, including the paradigmatic and experimentally relevant Lieb-Liniger model, this bound is satisfied for any physical initial state we are aware of, including local thermal states and zero-entropy states with non-differentiable occupation functions. Contrary to generic hyperbolic conservation equations, which typically develop entropy-increasing shocks in finite time, our work establishes that no such shocks appear in the GHD equations of those models. Consequently, entropy is conserved for all times.

The proof is based on a novel technique, outlined in our companion paper~\cite{hübner2024newquadraturegeneralizedhydrodynamics}, where the GHD equation is mapped onto a functional fixed-point problem, with space and time appearing as external parameters. For the models considered in this paper, this fixed-point problem is contracting in the Banach space of bounded functions, implying existence and uniqueness of the fixed point. In other models, particularly those with non-integrable scattering shift, the fixed-point equation still holds but is not necessarily contracting in this Banach space. However, this does not imply non-uniqueness of solutions. The simplest example of this is the hard rods model, discussed at the end of Section \ref{ssectheightfielddiscussion}. Other techniques to establishing existence and uniqueness might be necessary, for instance by considering a different Banach space or by using some other variant of the fixed-point equation (see for instance the version in \cite{hübner2024generalizedhydrodynamicsintegrablequantum} Appendix A, which was used numerically to overcome the non-convergence of the fixed-point iteration). It would be interesting to investigate these questions to other integrable models, in particular soliton gas models, which have non-integrable scattering shifts~\cite{el2011kinetic}. In this context, it is known that special initial states (finite component states) give rise to finite dimensional hyperbolic conservation laws in which shock formation is absent~\cite{Pavlov2012}. Therefore, it is natural to expect the same for general initial states, however, it is still an open question.

Apart from the bound on the occupation function, the proof requires, if the bare velocity $v(p)$ grows in amplitude with the momentum, that the density decays sufficiently fast as $p\to \pm \infty$. In a nutshell, this condition is needed to ensure that the effective velocity is finite. It would be interesting to study in the case where the decay is slower, whether the effective velocity could indeed diverge and what the consequences would be.

Having established the GHD equation in a weak form may be useful for proofs of the emergence of the GHD equation in specific models. Indeed, weak forms are easier to establish; for instance, the integral form \eqref{def:integral} only requires the proof that currents tend to their averages within local maximal entropy states sufficiently uniformly in space and time, and from Lemma \ref{lem:equivalenceweak} this implies the potential form for which we have shown existence and uniqueness.

Furthermore, our proof only regards the central equation of the GHD framework, the Euler equation. We believe that is should be possible to apply similar ideas to other aspects, for instance the evolution of correlation functions and the higher-order corrections to GHD. Particularly interesting and experimentally relevant is the GHD equation in the presence of an external trapping potential~\cite{10.21468/SciPostPhys.2.2.014}. In this case turbulent-like behavior has been observed numerically~\cite{PhysRevResearch.6.023083}, which suggest the formation of higher order discontinuities. However, our fixed-point problem approach for now cannot be applied to this case.

\backmatter
\bmhead{Acknowledgments} The authors thank T. Bonnemain and E.V. Ferapontov for discussions. BD was supported by the Engineering and Physical Sciences Research Council (EPSRC) under grant EP/W010194/1. FH acknowledges funding from the faculty of Natural, Mathematical \& Engineering
Sciences at King’s College London. 

\bmhead{Data availability} The manuscript has no associated data.

\begin{appendices}

\section{Equivalence between weak formulations}
\label{appweak}

We provide the proof of Lemma \ref{lem:equivalenceweak}. It will be convenient to introduce yet another, intermediate weak form. By integrating the potential form, it is relatively simple to re-state the potential form as a global conservation form.
\begin{definition}[Global conservation form]\label{def:global} We say that a function $n\in B\ind{T}(\mathbb R^3)\cap B(\mathbb R^3,\abs{v})$ satisfies the GHD equation in partial global conservation form, if for every $x_1,x_2,p\in\mathbb R$ the function $t\mapsto Q(t,x_1,x_2,p):= \int_{x_1}^{x_2} \dd{x} n(t,x,p)1\upd{dr}(t,x,p)$ is absolutely continuous on a dense subset $\tau(x_1,x_2,p)\subset\mathbb R$ uniformly in $x_1,x_2$, and for every $t_1,t_2,p\in\mathbb R$, the function $x\mapsto J(t_1,t_2,x,p)= \int_{t_1}^{t_2} \dd{t} n(t,x,p)v\upd{dr}(t,x,p)$ is absolutely continuous on a dense subset $\chi(t_1,t_2,p)\subset\mathbb R$ uniformly in $t_1,t_2$, and, denoting $\check Q(t,x_1,x_2,p)$ and $\check J(t_1,t_2,x,p)$ their unique continuous extensions,
\begin{equation}\label{equ:integralformQJ}
    \check Q(t_2,x_1,x_2,p)-\check Q(t_1,x_1,x_2,p)
    +\check J(t_1,t_2,x_2,p) - \check J(t_1,t_2,x_1,p)=0
\end{equation}
for every $p\in\mathbb R$ and every $(x_1,x_2,t_1,t_2)\in\mathbb R^4$. It is in complete conservation form if $\tau(x_1,x_2,p)= \chi(t_1,t_2,p)=\mathbb R$.
\end{definition}
Note the GHD equation in integral form manifestly follows from the global conservation form in their complete variants. The former also follows from the latter in their partial variants because the set $\{(x_1,x_2,t_1,t_2):x_1,x_2\in\mathbb R,\,t_1,t_2\in\tau(x_1,x_2,p)\}$ is dense in $\mathbb R^4$, and likewise for the set $\{(x_1,x_2,t_1,t_2):t_1,t_2\in\mathbb R,\,x_1,x_2\in\chi(x_1,x_2,p)\}$, thus so is their intersection; and on this intersection \eqref{equ:integralformQJ} gives \eqref{equ:integralform}. Note also that \eqref{equ:integralformQJ} in fact follows from the statement \eqref{equ:integralform} along with the absolutely continuity statements of the global conservation form. We will show
\begin{lemma}[Equivalence of weak statements]\label{lem:equivalenceweakapp}
Let $n\in B\ind{T}(\mathbb R^3)\cap B(\mathbb R^3,\abs{v})$. Then concerning the GHD equation for $n$: complete potential form $\Leftrightarrow$ complete integral form $\Leftrightarrow$ complete global conservation form; and partial potential form $\Leftrightarrow$ partial integral form $\Leftrightarrow$ partial global conservation form $\Leftrightarrow$ weak form.
\end{lemma}

In order to lighten the notation, we use $n1\upd{dr}(t,x,p) := n(t,x,p)1\upd{dr}(t,x,p)$ and $nv\upd{dr}(t,x,p) := n(t,x,p)v\upd{dr}(t,x,p)$. We recall that the results are in fact very general: we only need to use uniform boundedness of the density $n1\upd{dr}$ and current $nv\upd{dr}$ of the conservation law: $n1\upd{dr},\,nv\upd{dr}\in B(\mathbb R^3)$. Further, $p$ only arises as a parameter, so we fix it to some real number throughout; results are valid for every $p\in\mathbb R$.

\begin{proof} Clearly, all partial variants follow from the complete variants. For simplicity we consider only partial variants and omit the term ``partial"; the proofs of equivalence of complete variants is obtained from the statement just below Definition \ref{def:global}, and by straightforward adjustments of parts I and IV of this proof.

{\bf I.} Let us show potential form $\Leftrightarrow$ global conservation form. In one direction, given the potential form, we define $\check Q(t,x_1,x_2,p) = N(t,x_2,p)-N(t,x_1,p)$ for every $x_1,x_2$, and $\check J(t_1,t_2,x,p) = N(t_2,x,p)-N(t_1,x,p)$ for every $t_1,t_2$. By the conditions of the potential form, the former is absolutely continuous in $t$ uniformly in $x_1,x_2$, and the latter in $x$ uniformly in $t_1,t_2$. Further, using \eqref{equ:Ndifferentialequations}, for a.e.~$t$ we have $\check Q(t,x_1,x_2,p) = \int_{x_1}^{x_2} \dd{x} n1\upd{dr}(t,x,p)$, and for a.e.~$x$, $\check J(t_1,t_2,x,p) = \int_{t_1}^{t_2} \dd{t} nv\upd{dr}(t,x,p)$. Finally, we may write
\begin{equation}\begin{aligned}
	&\big(N(t_2,x_2,p)-N(t_2,x_1,p)\big)
	-
	\big(N(t_1,x_2,p)-N(t_1,x_1,p)\big)\\
	&\hspace{2cm}
	-
	\big(N(t_2,x_2,p)-N(t_1,x_2,p)\big)
	+
	\big(N(t_2,x_1,p)-N(t_1,x_1,p)\big)
	=0
	\end{aligned}
\end{equation}
for all $x_1,x_2,t_1,t_2$. Thus we have the global conservation form \eqref{equ:integralformQJ}. In the other direction, given the global conservation form, we define $N(t,x,p) = \check Q(0,0,x,p) + \check J(0,t,x,p) = \check J(0,t,0,p) + \check Q(t,0,x,p)$. By using either of these expressions, this satisfies either of the equations in \eqref{equ:Ndifferentialequations}, thus we get the potential form.

{\bf II.} We now show potential form $\Rightarrow$ weak form. Assume the potential form and let $f:\mathbb R^2\to\mathbb R$ be a rapidly decreasing function. Then
\begin{align}
    &\int \dd{x}\dd{t} \partial_t f(t,x)n1\upd{dr}(t,x,p)
    =\int \dd{x}\dd{t} \partial_t f(t,x)\partial_x N(t,x,p)\nonumber\\
    &=-\int \dd{x}\dd{t} \partial_x \partial_t f(t,x) N(t,x,p)
    =\int \dd{x}\dd{t} \partial_x f(t,x) \partial_t N(t,x,p)\nonumber\\
    &=-\int \dd{x}\dd{t} \partial_x f(t,x)
    nv\upd{dr}(t,x,p)
\end{align}
where in the first, third and fourth equality we used a.e.~differentiability of  $N(t,x,p)$ and the derivatives being given by \eqref{equ:Ndifferentialequations}, and in the second equality, the bound $\abs{N(t,x,p)}\leq \abs{x} \norm{n1\upd{dr}}_\infty$ coming from integrating the first equation in \eqref{equ:Ndifferentialequations}.

{\bf III.} We now show weak form $\Rightarrow$ integral form.  Assume the weak form. Let $(x_1,x_2,t_1,t_2)\in\mathbb R^4$ with $x_1<x_2$ and $t_1<t_2$. There is a sequence of rapidly decreasing functions $f_j$ such that
\begin{equation}
	f_j \to f: (t,x)\mapsto \left\{\begin{array}{ll}
	t-t_1 & \big(t\in [t_1,t_1+\epsilon],\,x\in[x_1+(t-t_1),x_2-(t-t_1)]\big)\\
	t_2-t & \big(t\in [t_2-\epsilon,t_2],\,x\in[x_1+(t_2-t),x_2-(t_2-t)]\big)\\
	x-x_1 & \big(x\in [x_1,x_1+\epsilon],\,t\in [t_1+(x-x_1),t_2-(x-x_1)]\big)\\
	x_2-x & \big(x\in[x_2-\epsilon,x_2],\,t\in [t_1+(x_2-x),t_2-(x_2-x)]\big)\\
	0 & \mbox{(otherwise)}
	\end{array}\right.
\end{equation}
uniformly. Note that $f$ is continuous, and has piecewise constant derivatives supported on a tubular neighbourhood of width $\epsilon$ lying inside, and touching the boundary of, the rectangle defined by the opposite corners $(t_1,x_1)$ and $(t_2,x_2)$. By the dominated convergence theorem, the limit on $j$ of the integral \eqref{equ:weakform} can be evaluated by taking the limit on the integrand. Thus we use this $f$ in \eqref{equ:weakform}, where the integral is therefore supported on the tubular neighbourhood. We evaluate $\epsilon^{-1}$ times the result in the limit $\epsilon\to0$. The ``corner" regions of integration $\{(t_a+s,x_b+y):(x,y)\in [-\epsilon,\epsilon]^2,a,b\in[1,2]\}$ give a contribution $O(\epsilon)$ and can be neglected. We are left with four pieces: the four edges of the rectangle. The smallest-time spatial edge gives
\begin{equation}
	\epsilon^{-1}\int_{t_1}^{t_1+\epsilon} \dd{t}
    \int_{x_1+\epsilon}^{x_2-\epsilon} \dd{x} n1\upd{dr}(t,x,p)
	=
	\epsilon^{-1}\int_{t_1}^{t_1+\epsilon} \dd{t} \Bigg(
    \int_{x_1}^{x_2} \dd{x} n1\upd{dr}(t,x,p)
    \Bigg)
	+ O(\epsilon).
\end{equation}
The result inside the $t$ integral exists for every $x_1,x_2,t$ and is bounded on $t\in\mathbb R$. Thus the result of $\lim_{\epsilon\to0^+}$ exists for every $x_1,x_2$ and a.e.~$t_1$ and gives $\int_{x_1}^{x_2} \dd{x}n1\upd{dr}(t_1,x,p)$ (because the function $t_1\mapsto \int_{0}^{t_1} \dd{t} \int_{x_1}^{x_2}\dd{x} n1\upd{dr}(t,x,p)$ is absolutely continuous with that derivative). Summing similar results for the other edges, this gives \eqref{equ:integralform} for a.e.~$(x_1,x_2,t_1,t_2),\,x_1<x_2,\,t_1<t_2$.

{\bf IV.} Finally, we show integral form $\Rightarrow$ global conservation form. We first establish that for every $x_1,x_2\in\mathbb R,\,x_1<x_2$, there is a dense subset $\tau(x_1,x_2,p)\subset\mathbb R$ such that $\tau(x_1,x_2,p)\ni t\mapsto \int_{x_1}^{x_2} \dd{x}n1\upd{dr}(t,x,p)\in\mathbb R$ is absolutely continuous. By the above result, for a.e. $(x_1,x_2,t_2)\in\mathbb R^3,\,x_1<x_2$, \eqref{equ:integralform} holds for a.e.~$t_1\in\mathbb R,\,t_1<t_2$. Call the former set $\Omega$, and the latter set $\tau_1(x_1,x_2,t_2,p)$. As it is dense, for every $\epsilon>0$ there is $t_1,t_1'\in \tau_1(x_1,x_2,t_2,p)$ with $\abs{t_1'-t_1}<\epsilon$. The difference of the left-hand side of \eqref{equ:integralform} for $(x_1,x_2,t_1',t_2)$ and for $(x_1,x_2,t_1,t_2)$ is
\begin{equation}
    \int_{x_1}^{x_2} \dd{x}n1\upd{dr}(t_1,x,p) -
    \int_{x_1}^{x_2} \dd{x}n1\upd{dr}(t_1',x,p) + O(\epsilon)
\end{equation}
by boundedness of the integrands in the $t$ integrals, uniformly for $(x_1,x_2,t_2)\in\Omega$. As this difference is zero, this establishes Lipschitz continuity, hence absolute continuity, of $\int_{x_1}^{x_2} \dd{x}n1\upd{dr}(t_1,x,p)$ in $t_1\in \tau_1(x_1,x_2,t_2,p)$, uniformly. As the result does not depend on $t_2$, we may take $\tau(x_1,x_2,p) = \cup_{t_2\in\mathbb R: (x_1,x_2,t_2)\in\Omega} \tau_1(x_1,x_2,t_2,p)$ which is dense in $\mathbb R$, and we conclude that absolute continuity in $t_1$ holds for a.e.~$x_1<x_2$ uniformly. By continuity in $x_1,x_2$ and uniformity it holds for all $x_1<x_2$.

Similar arguments show that for every $t_1,t_2\in\mathbb R,\,t_1<t_2$, there is a dense subset $\chi(t_1,t_2,p)\in\mathbb R$ such that $\chi(t_1,t_2,p)\ni x\mapsto \int_{t_1}^{t_2} \dd{t}nv\upd{dr}(t,x,p)\in\mathbb R$ is absolutely continuous. Other orderings $x_1>x_2$ and $t_1>t_2$ are obtained by changing the signs of the integrals.
\end{proof}

\end{appendices}

\bibliography{refs.bib}

\end{document}